\ifx\DCGVer\undefined
   \newcommand{\InDCGVer}[1]{}%
   \newcommand{\InNotDCGVer}[1]{#1}%
\else
   \newcommand{\InDCGVer}[1]{#1}%
   \newcommand{\InNotDCGVer}[1]{}%
\fi

\InDCGVer{%
   \documentclass[smallextended]{svjour3}%
}%
\InNotDCGVer{%
   \documentclass[12pt]{article}%
}

\usepackage{graphicx}
\usepackage[active]{srcltx} % comment out if you do not have it
\usepackage{paralist}%
\usepackage[cm]{fullpage}%
\usepackage{color}%
\usepackage{xspace}%
\usepackage{amsmath}%
\usepackage{amsfonts}%  Real number set notation
\usepackage{euscript}%
\usepackage{mleftright}%
\usepackage{caption}%
\InNotDCGVer{%
   \usepackage{titlesec}%
   \titlelabel{\thetitle. }%
}

\usepackage[pdfpagelabels,plainpages=false,hypertexnames=false]{hyperref}%
\hypersetup{%
   breaklinks,%
   ocgcolorlinks,
   colorlinks=true,%
   linkcolor=[rgb]{0.45,0.0,0.0},%
   urlcolor=[rgb]{0.05,0.390,0.0},%
   citecolor=[rgb]{0,0,0.45}
}%

\numberwithin{figure}{section}
\numberwithin{table}{section}
\numberwithin{equation}{section}%

\InNotDCGVer{%
   \usepackage[amsmath,thmmarks]{ntheorem}
   \theoremseparator{.}%
}%

\newlength{\savedparindentx}

\InNotDCGVer{%
   \theoremstyle{plain}%
   \newtheorem{theorem}{Theorem}[section] % section
   \newtheorem{lemma}[theorem]{Lemma}

   \theoremstyle{plain}%
   \theoremheaderfont{\sf}
   \theorembodyfont{\upshape}%

   \newtheorem{definition}[theorem]{Definition}%
   \newtheorem{assumption}[theorem]{Assumption}%
   
   \newtheorem*{remark:unnumbered}{Remark}%
   \newtheorem{remark}[theorem]{Remark}%
   \newtheorem{remark:ext}[theorem]{Remark}%

}%

\InNotDCGVer{%   
\newcommand{\myqedsymbol}{\rule{2mm}{2mm}} 

\theoremheaderfont{\em}%
\theorembodyfont{\upshape}%
\theoremstyle{nonumberplain}
\theoremseparator{}
\theoremsymbol{\myqedsymbol}
\newtheorem{proof}{Proof:}%
}

\InDCGVer{%
   \spnewtheorem*{remark:unnumbered}{Remark}{\itshape}{\rmfamily}
   \spnewtheorem{assumption}{Assumption}{\itshape}{\rmfamily}

   \smartqed   
}

\newcommand{\cardin}[1]{\left| {#1} \right|}
\newcommand{\pbrc}[1]{\mleft[ {#1} \mright]}%
\newcommand{\vrt}{\Mh{v}}%
\newcommand{\SubCurve}[3]{#1\pbrc{#2, #3}}

\definecolor{blue25}{rgb}{0.0,0,0.85}%
\newcommand{\emphic}[2]{%
   \textcolor{blue25}{%
      \textbf{\emph{#1}}}%
   \index{#2}}

\newcommand{\emphi}[1]{\emphic{#1}{#1}}%

\newcommand{\Term}[1]{\textsf{#1}}

\newcommand{\NPHard}{\Term{NP-Hard}\xspace}
\newcommand{\NP}{\Term{NP}\xspace}

\providecommand{\DFS}{\Term{DFS}\xspace}

\newcommand{\leashChar}{\Mh{\ell}}
\newcommand{\leashX}[1]{\leashChar\pth{#1}}

\newcommand{\leashXM}[1]{\leashChar^-\pth{#1}}
\newcommand{\leashXP}[1]{\leashChar^+\pth{#1}}

\newcommand{\morphChar}{\Mh{\mu}}%
\newcommand{\morphX}[1]{\morphChar\pth{#1}}%

\newcommand{\height}{\mathsf{height}}
\newcommand{\heightX}[1]{\height\pth{#1}}

\newcommand{\hhX}[2]{\mathsf{h}\pth{#1, #2}}
\newcommand{\hhXX}[4]{\mathsf{h}\pth{#1, #2, #3, #4}}
\newcommand{\HHopt}{\mathsf{h}_{\mathrm{opt}}}

\newcommand{\Frechet}{Fr\'{e}chet\xspace}
\newcommand{\distChar}{\mathsf{d}}%
\newcommand{\distZ}[3]{\distChar_{#1}\pth{#2, #3}}
\newcommand{\dist}[2]{\distChar\pth{#1,#2}}%
\newcommand{\MakeBig}{\rule[-.2cm]{0cm}{0.4cm}}

\newcommand{\etal}{\textit{et~al.}\xspace}
\newcommand{\brc}[1]{\left\{ {#1} \right\}}

\newcommand{\Polytope}{\Mh{\mathcal{P}}}%
\newcommand{\PolySet}{\Mh{\EuScript{Q}}}%
\newcommand{\Polygon}{\Mh{{Q}}}%
\newcommand{\bpnt}{\Mh{b}}%

\newcommand{\distP}[1]{\Mh{d}_{\Polytope}\pth{#1}}%
\newcommand{\MAxis}{\Mh{M}}%

\newcommand{\spnt}{\Mh{s}}%
\newcommand{\pnt}{\Mh{p}}%
\newcommand{\pntA}{\Mh{q}}

\newcommand{\PntSet}{\Mh{P}}

\renewcommand{\th}{th\xspace}

   \newcommand{\HLinkShort}[2]{\hyperref[#2]{#1\ref*{#2}}}
   \newcommand{\HLink}[2]{\hyperref[#2]{#1~\ref*{#2}}}
   \newcommand{\HLinkPage}[2]{\hyperref[#2]{#1~\ref*{#2}%
         $_\text{p\pageref{#2}}$}}
   \newcommand{\HLinkPageOnly}[1]{\hyperref[#1]{Page~\refpage*{#1}%
         $_\text{p\pageref{#1}}$}}
   \newcommand{\HLinkSuffix}[3]{\hyperref[#2]{#1\ref*{#2}{#3}}}
   \newcommand{\HLinkPageSuffix}[3]{\hyperref[#2]{#1\ref*{#2}%
         #3$_\text{p\pageref{#2}}$}}%

\newcommand{\figlab}[1]{\label{fig:#1}}
\newcommand{\figref}[1]{\HLink{Figure}{fig:#1}}

\newcommand{\seclab}[1]{\label{sec:#1}}
\newcommand{\secref}[1]{\HLink{Section}{sec:#1}}

\providecommand{\deflab}[1]{\label{def:#1}}
\newcommand{\defref}[1]{\HLink{Definition}{def:#1}}

\newcommand{\apndlab}[1]{\label{apnd:#1}}
\newcommand{\apndref}[1]{\HLink{Appendix}{apnd:#1}}

\newcommand{\itemlab}[1]{\label{item:#1}}
\newcommand{\itemref}[1]{\HLinkSuffix{(}{item:#1}{)}}

\newcommand{\remlab}[1]{\label{rem:#1}}
\newcommand{\remref}[1]{\HLink{Remark}{rem:#1}}%

\providecommand{\eqlab}[1]{}%
\renewcommand{\eqlab}[1]{\label{equation:#1}}
\newcommand{\Eqref}[1]{\HLinkSuffix{Eq.~(}{equation:#1}{)}}

\newcommand{\asmlab}[1]{\label{assumption:#1}}
\newcommand{\asmref}[1]{\HLink{Assumption}{assumption:#1}}%

\newcommand{\lemlab}[1]{\label{lemma:#1}}
\newcommand{\lemref}[1]{\HLink{Lemma}{lemma:#1}}%

\newcommand{\thmlab}[1]{{\label{theo:#1}}}
\newcommand{\thmref}[1]{\HLink{Theorem}{theo:#1}}

\providecommand{\Mh}[1]{#1}% 
\providecommand{\pth}[1]{\mleft({#1}\mright)}

\newcommand{\distFrHChar}{{\mathsf{d}_{\mathcal{H}}}}
\newcommand{\distFrH}[2]{\distFrHChar\pth{#1, #2}}

\newcommand{\distFrChar}{\mathsf{d}_{\mathcal{F}}}
\newcommand{\distFr}[2]{\distFrChar\pth{#1, #2}}

\newcommand{\vTop}{\Mh{\mathsf{t}}}
\newcommand{\vBot}{\Mh{\mathsf{b}}}

\newcommand{\vTopL}{\mathsf{t}_l}
\newcommand{\vTopR}{\mathsf{t}_r}
\newcommand{\vBotL}{\mathsf{b}_l}
\newcommand{\vBotR}{\mathsf{b}_r}

\newcommand{\CTop}{\Mh{\mathsf{T}}}%
\newcommand{\CBot}{\Mh{\mathsf{B}}}%
\newcommand{\CLeft}{\Mh{\mathsf{L}}}%
\newcommand{\CMid}{\Mh{\mathsf{M}}}%
\newcommand{\CRight}{\Mh{\mathsf{R}}}%
\newcommand{\EdgeSet}{\Mh{\mathcal{E}}}%
\newcommand{\EdgeSetS}{\Mh{S}}%

\newcommand{\Disk}{\Mh{\mathcal{D}}}%

\newcommand{\Pocket}{\mathsf{P}}%
\newcommand{\bd}{\partial}%
\newcommand{\face}{F}%

\newcommand{\TVertices}{\mathsf{M}}
\newcommand{\TLeftC}{\TVertices_l}
\newcommand{\TRightC}{\TVertices_r}
\newcommand{\TLeft}[1]{\TLeftC\pth{#1}}
\newcommand{\TRight}[1]{\TRightC\pth{#1}}

\newcommand{\pTop}{\mathsf{t}}
\newcommand{\pBot}{\mathsf{b}}

\newcommand{\walk}{\omega}
\newcommand{\walkA}{\chi}

\newcommand{\widthX}[1]{\mathrm{width}\pth{#1}}

\newcommand{\lenX}[1]{\left\|#1\right\|}

\newcommand{\edge}{\Mh{{e}}}%

\newcommand{\Vertices}[1]{\mathsf{V}\pth{#1}}

\newcommand{\BPath}{\sigma}
\newcommand{\PathL}{\sigma_l}
\newcommand{\PathA}{\Mh{\zeta}}%
\newcommand{\PathB}{\Mh{\psi}}%

\newcommand{\leftC}[1]{\pi_\mathrm{L,{#1}}}
\newcommand{\rightC}[1]{\pi_\mathrm{R,{#1}}}

\newcommand{\iClass}{\mathsf{h}}
\newcommand{\iClassOpt}{\mathsf{h}_{\mathrm{opt}}}

\newcommand{\signXL}[1]{\Mh{\mathrm{sg}}\pth{#1}}%

\newcommand{\Graph}{\mathsf{G}}

\renewcommand{\Re}{{\rm I\!\hspace{-0.025em} R}}

\newcommand{\eps}{\varepsilon}

\newcommand{\IncludeGraphics}[2][]{%
   \typeout{}%
   \typeout{Graphics: #2}%
   \typeout{\ includegraphics[#1]{#2}}%
   \includegraphics[#1]{#2}
   \typeout{}%
}

\newcommand{\concat}{\cdot}

\newcommand{\strip}{\Mh{C}}%
\newcommand{\chunk}{\Mh{C'}}%

\newlength{\savedparindent}
\newcommand{\SaveIndent}{\setlength{\savedparindent}{\parindent}}
\newcommand{\RestoreIndent}{\setlength{\parindent}{\savedparindent}}

\newcommand{\dLeft}{\mathsf{d}_{\CLeft}}
\newcommand{\dBoth}{\mathsf{d}}

\newcommand{\CostL}[1]{\mathrm{cost}\pth{#1}}
\newcommand{\CostSetL}[1]{\mathrm{Cost}\pth{#1}}

\newcommand{\tall}{\ensuremath{\tau}\xspace}

\newcommand{\ds}{\displaystyle}

\usepackage{tikz}  % Used for the extended footnotes if needed
\usepackage{pifont}%
\usepackage[symbol*]{footmisc}

\newcommand{\si}[1]{#1}

\newcommand*\circled[1]{\footnotesize\tikz[baseline=(char.base)]{
    \node[shape=circle,draw,inner sep=0.2pt] (char) {#1};}}

\DefineFNsymbols*{stars}[text]{%
   {\ding{172}} %
   {\ding{173}} %
   {\ding{174}} %
   {\ding{175}} %
   {\ding{176}} %
   {\ding{177}} %
   {\ding{178}} %
   {\ding{179}} %
   {\ding{180}} %
   {{\protect\circled{{\tiny 10}}}}%
   {{\protect\circled{{\tiny 11}}}}%
   {{\protect\circled{{\tiny 12}}}}%
   {{\protect\circled{{\tiny 13}}}}%
   {{\protect\circled{{\tiny 14}}}}%
   {{\protect\circled{{\tiny 15}}}}%
   {{\protect\circled{{\tiny 16}}}}%
   {{\protect\circled{{\tiny 17}}}}%
   {{\protect\circled{{\tiny 18}}}}%
   {{\protect\circled{{\tiny 19}}}}%
   {{\protect\circled{{\tiny 20}}}}%
   {{\protect\circled{{\tiny 21}}}}%
   {{\protect\circled{{\tiny 22}}}}%
}%
\newcommand{\Tree}{\Mh{\mathcal{T}}}%

\newcommand{\curveA}{\mathsf{B}}%
\newcommand{\curveB}{\mathsf{C}}%

\newcommand{\pntOnA}{\mathsf{b}}%
\newcommand{\pntOnB}{\mathsf{c}}%

\newcommand{\PathSet}{\Mh{\Pi}}%
\newcommand{\PSetX}[1]{\Mh{\Pi}_{#1}}%

\IfFileExists{sariel_computer.sty}{\def\sarielComp{1}}{}
\ifx\sarielComp\undefined%
\newcommand{\SarielComp}[1]{}
\newcommand{\NotSarielComp}[1]{#1}%
 \else
\newcommand{\SarielComp}[1]{#1}%
\newcommand{\NotSarielComp}[1]{}%
\fi

\SarielComp{%
   \ifx\colorMath\undefined%
   \else
   \DefineNamedColor{named}{ColorMath}{rgb}{0.5,0.2,0}
       \renewcommand{\Mh}[1]{{\textcolor{ColorMath}{#1}}}
   \fi
}

\begin{document}

\title{How to Walk Your Dog in the Mountains with No Magic Leash%
   \thanks{A preliminary version of this paper appeared in SoCG 2012
      \cite{hnss-hwydm-12}.}%
}

\author{Sariel Har-Peled%
   \InNotDCGVer{%
      \thanks{Department of Computer Science, University of Illinois,
         Urbana-Champaign; \url{sariel@illinois.edu}. %
         Work on this paper was partially supported by NSF AF awards
         CCF-0915984, % Historical grant
         CCF-1421231, and % Started June 2014
         CCF-1217462.  % Started June 2012
      }%
   }%
   \and%
   Amir Nayyeri%
   \InNotDCGVer{%
      \thanks{School of Electrical Engineering \& Computer Science,
         Oregon State University; \url{nayyeria@eecs.oregonstate.edu}.
         Part of this work was done while visiting Toyota
         Technological Institute at Chicago; %
         \url{nayyeri2@illinois.edu}.%
      }%
   }%
   \and%
   Mohammad Salavatipour%
   \InNotDCGVer{%
      \thanks{Department of Computing Science, University of Alberta;
         Supported by NSERC and Alberta Innovates.  Part of this work
         was done while visiting Toyota Technological Institute at
         Chicago; \url{mreza@cs.ualberta.ca}.%
      }%
   }%
   \and%
   Anastasios Sidiropoulos%
   \InNotDCGVer{%
      \thanks{%
         Department of Computer Science \& Engineering, and Department
         of Mathematics, The Ohio State University;
         \url{sidiropoulos.1@osu.edu}.%
      }%
   }%
}%

\InDCGVer{%
   \institute{S. Har-Peled%
      \at%
      Department of Computer Science, University of Illinois,
      Urbana-Champaign; \url{sariel@illinois.edu}. Work on this paper
      was partially supported by NSF AF awards CCF-0915984,
      CCF-1421231, and % Started June 2014
      CCF-1217462.  % Started June 2012
      \and%
      A. Nayyeri %
      \at School of Electrical Engineering \& Computer Science, Oregon
      State University; \url{nayyeria@eecs.oregonstate.edu}. %
      \and %
      M. Salavatipour %
      \at Department of Computing Science, University of Alberta;
      Supported by NSERC and Alberta Innovates.  Part of this work was
      done while visiting Toyota Technological Institute at Chicago;
      \url{mreza@cs.ualberta.ca}.%
      \and %
      A. Sidiropoulos %
      \at%
      Department of Computer Science \& Engineering, and Department of
      Mathematics, The Ohio State University;
      \url{sidiropoulos.1@osu.edu}.  Research supported in part by the
      NSF grants CCF 1423230 and CAREER 1453472. %
   }%
}

\maketitle

% Do not delete - Reset footnotes to numbers
\setfnsymbol{stars}

\begin{abstract}
    We describe a $O(\log n )$-approximation algorithm for computing
    the homotopic \Frechet distance between two polygonal curves that
    lie on the boundary of a triangulated topological disk.  Prior to
    this work, algorithms were known only for curves on the Euclidean
    plane with polygonal obstacles.
    
    A key technical ingredient in our analysis is a
    $O(\log n)$-approximation algorithm for computing the minimum
    height of a homotopy between two curves.  No algorithms were
    previously known for approximating this parameter.  Surprisingly,
    it is not even known if computing either the homotopic \Frechet
    distance, or the minimum height of a homotopy, is in \NP. %
    \InDCGVer{%
    }%
\end{abstract}

% \Sariel{Everybody: Please update your contact information!}

\section{Introduction}

Comparing the shapes of curves -- or sequenced data in general -- is a
challenging task that arises in many different contexts. The
\emphi{\Frechet{} distance} and its variants (e.g. dynamic
time-warping \cite{kp-sudtw-99}) have been used as a similarity
measure in various applications such as matching of time series in
databases \cite{kks-osmut-05}, comparing melodies in music information
retrieval \cite{sgh-csi-08}, matching coastlines over time
\cite{mdbh-cmpdfd-06}, as well as in map-matching of vehicle tracking
data \cite{bpsw-mmvtd-05,wsp-anmms-06}, and moving objects analysis
\cite{bbg-dsfm-08,bbgll-dcpcs-08}.  See \cite{ab-scfds-05,ag-cfdbt-95}
for algorithms for computing the \Frechet distance.

Informally, for a pair of such curves $f,g:[0,1]\to \Disk$, for some
ambient metric space $(\Disk, \distChar)$, their \Frechet distance is
the minimum length of a leash needed to traverse both curves in sync.  To
this end, imagine a person traversing $f$ starting from $f(0)$, and a
dog traversing $g$ starting from $g(0)$, both traveling continuously
along these curves without ever moving backwards.  Then the \Frechet
distance is the infimum over all possible traversals, of the maximum
distance between the person and the dog.  This notion can be
formalized via a reparameterization: a continuous bijection
$\phi:[0,1]\to[0,1]$.  The \emphi{width} of $\phi$, i.e., the longest
leash needed by $\phi$, is
$\displaystyle \widthX{\phi} = \sup_{t\in [0,1]} \dist{ \MakeBig
   f(t)}{g(\phi(t))}$,
where $\dist{x}{y}$ is the distance between $x$ and $y$ on $\Disk$.
Consequently, the \Frechet distance between $f$ and $g$ is defined to
be
\begin{align*}
    \distFr{f}{g}= \inf_{\phi:[0,1]\to [0,1]} \widthX{\phi},
\end{align*}
where $\phi$ ranges over all orientation-preserving homeomorphisms.

While this measure captures similarities between two curves when the
underlying space is Euclidean, it is not as informative for more
complicated underlying spaces such as a surface.  For example, imagine
walking a dog in the woods. The leash might get tangled as the dog and
the person walk on two different sides of a tree.  Since the \Frechet
distance cares only about the distance between the two moving points,
the leash would ``magically'' jump over the tree. In reality, when
there is no ``magic'' leash that jumps over a tree, one has to take
into account the extra length needed (for the leash) to pass over such
obstacles.

\paragraph{Homotopic \Frechet distance.}

To address this shortcoming, \emphi{homotopic \Frechet distance}, a
natural extension of the above notion was introduced by Chambers \etal
\cite{cdellt-wydwpt-10}.  Informally, revisiting the above person-dog
analogy, we consider the infimum over all possible traversals of the
curves, but this time, we require that the person is connected to the
dog via a leash, i.e.,~a curve that moves continuously over time.
Furthermore, one keeps track of the leash during the motion, where the
purpose is to minimize the maximum leash length needed.

To this end, consider a homotopy $h:[0,1]^2 \to \Disk$, which can also
be viewed as a homeomorphism between the unit square and $\Disk$.  For
parameters $\sigma,\tau \in [0,1]$ consider the one dimensional
functions $\leashChar(\tau) = h(\tau, \cdot):[0,1]\to\Disk$ and
$\morphChar(\sigma) = h(\cdot, \sigma):[0,1]\to\Disk$.  These are
parameterized curves that are the natural restrictions of $h$ into one
dimension.  We require that $\morphX{0} = f$ and $\morphX{1} = g$.
The \emphi{homotopy width} of $h$ is
$\ds \widthX{h} = \sup_{\tau\in [0,1]} \lenX{ \leashX{\tau} }$, and
the \emphi{homotopic \Frechet distance} between $f$ and $g$ is
\begin{align*}
    \distFrH{f}{g} = \inf_{h:[0,1]^2 \to \Disk} \widthX{h},
\end{align*}
where the infimum is over all homeomorphisms $h$ between $[0,1]^2$ and
$\Disk$, and $\lenX{\cdot}$ denotes the length of a curve.  Note that
$\leashX{\cdot}$, in particular, specifies a reparametrization between
the curves $f$ and $g$.

Clearly, $\distFrH{f}{g} \geq \distFr{f}{g}$ and, furthermore,
$\distFrH{f}{g}$ can be arbitrarily larger than $\distFr{f}{g}$.  We
remark that $\distFrH{f}{g} = \distFr{f}{g}$ for any pair of curves in
the Euclidean plane, as we can always pick the leash to be a straight
line segment between the person and the dog.  In other words, the map
$h$ in the definition of $\distFrHChar$ can be obtained from the map
$\phi$ in the definition of $\distFrChar$ via an appropriate affine
extension.  However, this is not true for general ambient spaces,
where the leash might have to pass over obstacles, hills, or more
generally regions of positive or negative curvature, etc. In
particular, in the general settings, usually, the leash would not be a
geodesic (i.e., a shortest path) during the motion.

\begin{figure}
    \centerline{%
       \begin{tabular}{\si{ccc}}
         \IncludeGraphics[page=1]{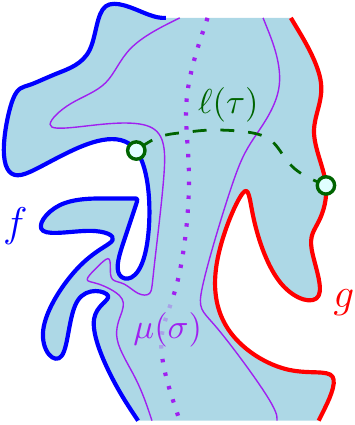}%
         &%
           \hspace{0.1\linewidth}
         &%
           \IncludeGraphics[page=2]{figs/homotopy}%
         \\
         (i) & & (ii)
       \end{tabular}%
    }
    \caption{(i) Two curves $f$ and $g$, and (ii) the parameterization
       of their homotopic \Frechet distance.}
    \figlab{morphing}
\end{figure}

The homotopic \Frechet distance is referred to as the \emphi{morph
   width} of $f$ and $g$, and it bounds how far a point on $f$ has to
travel to its corresponding point in $g$ under the morph of $h$
\cite{eghmm-nsmpa-02}. The length of $\morphX{\sigma}$ is the
\emph{height of the morph at time $\sigma$}, and the \emphi{height} of
such a morph is
$\heightX{h} = \sup_{\sigma\in[0,1] } \lenX{\morphX{\sigma}}$. The
\emphi{homotopy height} between $f$ and $g$, bounded by given starting
and ending leashes $\ell(0)$ and $\ell(1)$, is
\begin{align*}
    \hhXX{f}{g}{\Bigl.\ell(0)}{\ell(1)} = \inf_{h} \heightX{h},
\end{align*}
where $h$ varies over all possible morphs between $f$ and $g$, such
that each curve $\morphChar(\sigma)$ has one end on $\ell(0)$ and one
end on $\ell(1)$.  See \figref{morphing} for an example.  Note that if
we do not constrain the endpoints of the curves during the morph to
stay on $\ell(0)$ and $\ell(1)$, the problem of computing the minimum
height homotopy is trivial.  One can contract $f$ to a point, send it
to a point in $g$, and then expand it to $g$.  To keep the notation
simple, we use $\hhX{f}{g}$ when $f$ and $g$ have common endpoints.

Intuitively, the homotopy height measures how long the curve has to
become as it deforms from $f$ to $g$, and it was introduced by
Chambers and Letscher \cite{cl-hh-09,cl-ehh-10} and Brightwell and
Winkler~\cite{bw-sp-09}. Observe that if we are given the starting and
ending leashes $\leashX{0}$ and $\leashX{1}$ then the homotopy height
of $f$ and $g$ when restricted on homotopies that agree with
$\leashX{0}$ and $\leashX{1}$ is the homotopic \Frechet distance
between $\leashX{0}$ and $\leashX{1}$.

\medskip

Here, we are interested in the problems of computing the homotopic
\Frechet distance and the homotopy height between two simple polygonal
curves that lie on the boundary of an arbitrary triangulated
topological disk.

\paragraph{Why are these measures interesting?}
For the sake of the discussion here, assume that we know the starting
and ending leash of the homotopy between $f$ and $g$. The region
bounded by the two curves and these leashes, forms a topological disk,
and the mapping realizing the homotopic \Frechet distance specifies
how to sweep over $\Disk$ in a geometrically ``efficient'' way
(especially if the leash does not sweep over the same point more than
once), so that the leash (i.e., the sweeping curve) is never too long
\cite{eghmm-nsmpa-02}. As a concrete example, consider the two curves
as enclosing several mountains between them on the surface --
computing the homotopic \Frechet distance corresponds to deciding
which mountains to sweep first and in which order.

Furthermore, this mapping can be interpreted as surface
parameterization \cite{f-psast-97,ss-spmtf-00} and can thus be used in
applications such as texture mapping~\cite{bvig-psfnd-91,pb-stmss-00}.
In the texture mapping problem, we wish to find a continuous and
invertible mapping from the texture, usually a two-dimensional
rectangular image, to the surface.

Another interesting interpretation is when $f$ is a closed curve, and
$g$ is a point. Interpreting $f$ as a rubber band in a 3d model, the
homotopy height between $f$ and $g$ here is the minimum length the
rubber band has to have so that it can be collapsed to a point through
a continuous motion within the surface.  In particular, a short closed
curve with large homotopy height to any point in the surface is a
``neck'' in the 3d model.

To summarize, these measures seem to provide us with a fundamental
understanding of the structure of the given surface/model.

\paragraph{Continuous vs.~discrete.} 
Here we are interested in two possible models. In the
\emphi{continuous} settings, as described above, the leash moves
continuously in the interior of the domain. In the \emphi{discrete}
settings, the leash is restricted to the triangulation edges. As such,
a transition of the leash corresponds to the leash ``jumping'' over a
single face at each step. The two versions are similar in nature, but
technically require somewhat different tools and insights. This issue
is discussed more formally in \secref{prelims}.

\paragraph{Previous work.}

The problem of computing the (standard) \Frechet distance between two
polygonal curves in the plane has been considered by Alt and Godau
\cite{ag-cfdbt-95}, who gave a polynomial time algorithm. Eiter and
Mannila \cite{em-cdfd-94} studied the easier discrete version of this
problem.  Computing the \Frechet distance between surfaces
\cite{f-sdds-24}, appears to be a much more difficult task, and its
complexity is poorly understood.  The problem has been shown to be
\NPHard by Godau \cite{g-cmsbg-99}, while the best algorithmic result
is due to Alt and Buchin \cite{ab-scfds-05}, who showed that it is
upper semi-computable.

Efrat \etal \cite{eghmm-nsmpa-02} considered the \Frechet distance
inside a simple polygon as a way to facilitate sweeping it
efficiently. They also used the \Frechet distance with the underlying
geodesic metric as a way to obtain a morph between two curves. For
recent work on the \Frechet distance, see
\cite{cw-gfdis-10,cljl-ifd-11,%
   hr-fdre-11,cdhsw-cfdfp-11,dhw-afdrc-12,%
   cw-sppps-12} and references therein.

Chambers \etal \cite{cdellt-wydwpt-10} gave a polynomial time
algorithm to compute the homotopic \Frechet distance between two
polygonal curves on the Euclidean plane with polygonal obstacles.
Chambers and Letscher \cite{cl-hh-09,cl-ehh-10} and Brightwell and
Winkler \cite{bw-sp-09} considered the notion of minimum homotopy
height, and proved structural properties for the case of a pair of
paths on the boundary of a topological disk.  We remark that in
general, it is not known whether the optimum homotopy has polynomially
long description.  In particular, it is not known whether the problem
is in \NP.

Variants of the \Frechet distance for curves that are known to be
computationally hard, include (i) the problem of finding the most
similar \emph{simple} (i.e., no self crossings) curve to a given curve
on a surface \cite{sw-sce-13}, and (ii) computing the optimal \Frechet
distance when allowing shortcuts anywhere on one of the curves
\cite{bds-cfdsn-13}. Chambers and Wang \cite{cw-msbc2-13} study a
measure of similarity between curves that involves the minimum area
spanned by a homotopy.

For a Riemannian $2$-disk with boundary length $L$, diameter $d$ and
area $A \ll d$, Papasoglu \cite{p-ctd-13} showed that there is a
homotopy that contracts the disk to a point, such that the maximum
length of the homotopy curve is at most
\begin{math}
    L + 2d + O\bigl(\sqrt{A} \, \bigr).
\end{math}
Chambers and Rotman \cite{cr-clr2s-13} showed that given such a
homotopy with maximum length $L$ (i.e., any contraction of a disk to a
point), one can modify it into a homotopy using only the loops of a
base point $\pnt$, contracting the disk into $\pnt$, with maximum
length $L + 2d + \eps$, where $\eps >0$ is arbitrarily small.

\paragraph{Our results.}

In this paper, we consider the problems of computing the homotopic
\Frechet distance and the homotopy height between two simple polygonal
curves that lie on the boundary of a triangulated topological disk
$\Disk$ that is composed of $n$ triangles.  We give a polynomial time
$O(\log n)$-approximation algorithm for computing the homotopy height
between $f$ and $g$.  Our approach is based on a simple, yet delicate
divide and conquer approach.

We use the homotopy height algorithm as an ingredient for a
$O(\log n)$-approximation algorithm for the homotopic \Frechet
distance problem.  Here is a high-level description of our algorithm
for approximating the homotopic \Frechet distance: We first guess
(i.e., search over) the optimum (i.e.,~$\distFrH{f}{g}$).  Using this
guess, we classify parts of $\Disk$ as ``obstacles'', i.e., regions
over which a short leash cannot pass.  Consider the punctured disk
obtained from $\Disk$ after removing these obstacles.  Intuitively,
two leashes are homotopic if one can be continuously deformed to the
other within the punctured disk, while its endpoints remain on the
boundary during the deformation.  Observe that the leashes of the
optimum solution are homotopic.  We describe a greedy algorithm to
compute a ``small'' number of homotopy classes out of infinitely many
choices.  The homotopic \Frechet distance constrained to paths inside
one of these classes is a polynomial approximation to the homotopic
\Frechet distance in $\Disk$. We can then perform a binary search over
this interval to acquire a better approximation. An extended version
of the homotopy height algorithm is used in this algorithm in several
places.

The $O(\log n)$ factor shows up in the homotopic \Frechet distance
algorithm only because it uses the homotopy height as a subroutine.
Thus, any constant factor approximation algorithm for the homotopy
height problem implies a constant factor approximation algorithm for
the homotopic \Frechet distance.

We also shortly sketch, in \apndref{sweep:p}, an algorithm for
sweeping the boundary of a convex polytope in three dimensions, by a
guard that is connected by a continuously moving leash to a base
point on the boundary of the polytope. This algorithm is a warm-up
exercise for the more involved problem studied in this paper, and it
might be of independent interest.

\paragraph{Organization.}
We provide basic definitions in \secref{prelims}.  Then we consider
the discrete version of the homotopy height problem in
\secref{hh_discrete}. 
% This is how Chambers and Letscher formulated the problem.  
Later, in \secref{h:h:continuous}, we describe an
algorithm to approximately find the shortest homotopy in the
continuous settings.  In \secref{HFD}, we address the homotopic
\Frechet distance, for both the discrete and the continuous cases.  We
conclude in \secref{conclusions}.

% \begin{comment}
%%%%%%%%%%%%%%%%%%%%%%%%%%%%%%%%%%%%%%%%%%%%%%%%%%%%%%%%%%%%%%%%%%%%%     
\section{Background}
\seclab{prelims}

\subsection{Planar graphs}

Let $\Graph = (V,E)$ be a simple undirected graph with edge weights
$w:E \rightarrow \Re^+$.  For any $u,v\in V$ we denote by
$\distZ{\Graph}{x}{y}$ the shortest-path distance between $u$ and $v$
in $\Graph$, where every edge $e$ has length $w(e)$.  An
\emphi{embedding} of $\Graph$ in the plane maps the vertices of
$\Graph$ to distinct points in the plane and its edges to disjoint
paths except for the endpoints.  The faces of an embedding are maximal
connected subsets of the plane that are disjoint from the union of the
(images of the) edges of the graph.  The notation $\partial f$ refers
to the boundary of a single face $f$.  The term \emphi{plane graph}
refers to a graph together with its embedding in the plane.

The \emphi{dual} graph $\Graph^*$ of a plane graph $\Graph$ is the
(multi-)graph whose vertices correspond to the faces of $\Graph$,
where two faces are joined by a (dual) edge if and only if their
corresponding faces are separated by an edge of $\Graph$.  Thus, any
edge $e$ in $\Graph$ corresponds to a dual edge $e^*$ in $\Graph^*$,
any vertex $v$ in $\Graph$ corresponds to a face $v^*$ in $\Graph^*$
and any face $f$ in $\Graph^*$ corresponds to a vertex $f^*$ in
$\Graph^*$.

A \emphi{walk} $W$ in $\Graph$ is a sequence of vertices
$(v_1, v_2, \cdots ,v_k)$ such that each adjacent pair
$e_i = (v_{i},v_{i+1})$ is an edge in $\Graph$.  The length of $W$ is
$\lenX{W} = \sum_{i = 1}^{k-1}{w(e_i)}$.

Let $v_i$ and $v_j$ be two vertices that appear on $W$.  Here,
$W[v_i, v_j]$ denotes the sub-walk of $W$ that starts at the first
appearance of $v_i$ and ends at the first appearance of $v_j$ after
$v_i$ on $W$.  For two walks, $W_1 = (v_1, v_2, \dots, v_i)$ and
$W_2 = (v_i, v_{i+1}, \dots, v_j)$, their \emphi{concatenation} is
$W_1 \cdot W_2 = (v_1, v_2, \dots, v_i, v_{i+1}, \dots, v_j)$.

A walk with all the vertices being distinct is a \emphi{path}.  The
term $(u,v)$-walk refers to a walk that starts at $u$ and ends in $v$,
and $(u,v)$-path is defined similarly.  A walk is closed if its first
and last vertices are identical.  A closed path is a \emphi{cycle}.
Two walks \emph{cross} if and only if their images cross in the
plane. That is, no infinitesimal perturbation makes them disjoint.

\subsection{Piecewise linear surfaces and geodesics}

A \emphi{piecewise linear} surface is a $2$-dimensional manifold
composed of a finite number of Euclidean triangles by identifying
pairs of equal length edges.  In this paper, we work with piecewise
linear surfaces that can be embedded in three dimensional space such
that all triangles are flat and the surface does not cross itself.
This assumption lets us exploit existing shortest paths algorithms
for polyhedral surfaces\footnote{However, all of these algorithms
   should work verbatim even if the surface is not embedded in 3d,
   assuming it is an oriented and has the topology of a
   disk. Nevertheless, we keep this assumption to make the discussion
   more concrete and hopefully more intuitive.}.

% \begin{definition}
A triangulated surface is \emphi{non-degenerate} if no interior vertex
has curvature 0, i.e.,~when for every non-boundary vertex $x$, the sum
of the angles of the triangles incident to $x$ is not equal to $2\pi$.

\begin{assumption}
    \asmlab{non:degenerate}%
    We assume that the input surface is always non-degenerate.  One
    can turn any triangulated surface into being non-degenerate by
    perturbing all edge lengths by a factor of at most $1+\eps$, for
    some $\eps=O(1/n^2)$ (alternatively, one can perturb the vertices
    and edges).  This changes the metric of the surface by at most a
    factor of $1+1/n$, and thus the minimum height of a homotopy.
    Such a factor will be negligible for our approximation guarantee.
\end{assumption}
%\end{definition}

A \emphi{path} $\gamma$ on the surface $\Disk$ is a function
$\gamma:[0,1]\rightarrow \Disk$; $\gamma(0)$ and $\gamma(1)$ are the
endpoints of the path, and $\lenX{\gamma}$ denotes the length of
$\gamma$.  The path $\gamma$ is \emph{simple} if and only if it is
bijective.  A path is a \emphi{geodesic} if and only if it is locally
a shortest path; i.e., it cannot be shortened by an infinitesimal
perturbation.  In particular, global shortest paths are geodesics. The
terms path or \emphi{curve} are used interchangeably, and mean the
same thing.  A path or a curve is polygonal if it is composed of a
finite number of line segments.

\subsubsection{Computing shortest paths on a polyhedral surface}
\seclab{s:p:polytope}%

Mitchell \etal \cite{mmp-dgp-87} describe an algorithm to compute the
shortest path distance from a single source to the whole surface in
$O(n^2 \log n)$ time.  Underlying this algorithm are the following two
observations:
\smallskip%
\begin{compactenum}[\quad(A)]
    \item Shortest paths from the source $\spnt$ can not intersect
    in their interior. 

    \item Consider two points $\pnt$ and $\pnt'$ on an edge $\edge$ of
    $\Disk$, and their two shortest paths $\PathA$ and $\PathA'$,
    respectively, to the source $\spnt$. Furthermore, assume that
    these shortest paths approach $\edge$ from the same side. Then,
    all the shortest paths from $\spnt$ to the points on the edge
    $\edge$ between $\pnt$ and $\pnt'$ (coming from the same side of
    $\edge$ as $\PathA$ and $\PathA'$), must lie inside the disk on
    $\Disk$ having the boundary
    \begin{math}
        \spnt \concat \PathA \concat \edge[\pnt, \pnt'] \concat
        \PathA'.
    \end{math}
    This property requires that the topology of the input surface to
    be either a disk or, more generally, a punctured disk.
\end{compactenum}%
\smallskip%
These two observations still hold when the source is an edge instead
of a point.  As such, the algorithm of Mitchell
\etal~\cite{mmp-dgp-87} can be adapted to compute the shortest path
distance from an edge to the whole surface (with the same running
time). This requires modifying the wavefront maintenance and
propagation to be the distance from an edge instead from a point.

As such, by running this modified algorithm $O(n)$ times, one can
compute, in $O(n^3\log n)$ time, the shortest path from a set of
$O(n)$ edges to the whole surface.

\paragraph{Signature \& medial points.} 

Any shortest path in $\Disk$ is a polygonal line that intersects every
edge at most once and passes through a face along a segment.
Moreover, the shortest path crossing an edge looks locally like a
straight line segment, if one rotates the adjacent faces so that they
are coplanar.  See \cite{mmp-dgp-87} for more details.

Let $\EdgeSetS$ be a set of edges of $\Disk$, and let $\PathA$ be a
shortest path from $\EdgeSetS$ to a point $\pnt \in \Disk$.  The
\emphi{signature} of $\PathA$ is defined to be the ordered set of edges and
vertices (crossed or used) by $\PathA$.  Since $\PathA$ is locally a
straight line segment, we can rotate all faces that intersect $\PathA$
one by one so that $\PathA$ becomes a straight line.  It follows that
two geodesics with the same signature from $\pnt$ are identical.  A
point $\pnt$ on the surface is a \emphi{medial} point with respect to
$\EdgeSetS$ if there are more than one shortest paths (with different
signatures) from $\pnt$ to $\EdgeSetS$.

Note that a shortest path has a vertex of $\Disk$ in its interior, if
and only if, the vertex is a boundary vertex, or the vertex has
negative curvature (i.e., the total sum of the angles of the triangles
adjacent to this vertex is larger than $2\pi$). In particular, a
vertex with positive curvature (i.e., total sum of angles $< 2\pi$),
which is not on the boundary of $\Disk$, can not be in a signature of
a shortest path, see also \asmref{non:degenerate}.

\subsection{Homotopy and leash function}
\seclab{homotopy:continuous}%

Let $\CLeft$ and $\CRight$ be two paths with the same endpoints $s$
and $t$ on a surface $\Disk$.  A homotopy
$h:[0,1]\times [0,1] \rightarrow \Disk$ is a continuous function, such
that $h(\cdot,0) = \CLeft$, $h(\cdot,1) = \CRight$, $h(0,\cdot) = s$
and $h(1,\cdot) = t$.  So, for each $\tau \in [0,1]$, $h(\cdot,\tau)$
is an $(s,t)$-path.  The \emph{height} of such a homotopy (as defined
previously) is defined to be
$\sup_{\tau\in [0,1]}{\lenX{h(\cdot,\tau)}}$.

Let $\curveA$ and $\curveB$ be two disjoint curves.  A curve
connecting a point in $\curveA$ to a point in $\curveB$ is a
$(\curveA,\curveB)$-\emphi{leash}.  A $(\curveA,\curveB)$-\emphi{leash
   function} is a function $f$ that maps every $\tau \in [0,1]$ to a
leash with endpoints $\pntOnA(\tau)\in \curveA$ and
$\pntOnB(\tau) \in \curveB$ such that $\pntOnA:[0,1]\to \curveA$ and
$\pntOnB:[0,1] \to \curveB$ are reparametrizations of $\curveA$ and
$\curveB$, respectively.  A $(\curveA,\curveB)$-leash function $f$ is
\emph{continuous} if the leash $f(\tau)$ varies continuously with
$\tau$.  The \emph{height} of a leash function $f$ is
$\sup_{\tau\in[0,1]}{\lenX{f(\tau)}}$.  Recall that the \emph{\Frechet
   distance} between $\curveA$ and $\curveB$ is the height of the
minimum height $(\curveA,\curveB)$-leash function.  The
\emph{homotopic \Frechet distance} between $\curveA$ and $\curveB$ is
the height of the minimum height continuous $(\curveA,\curveB)$-leash
function.

\begin{figure}[t]
    \IncludeGraphics[page=1,scale=0.67]{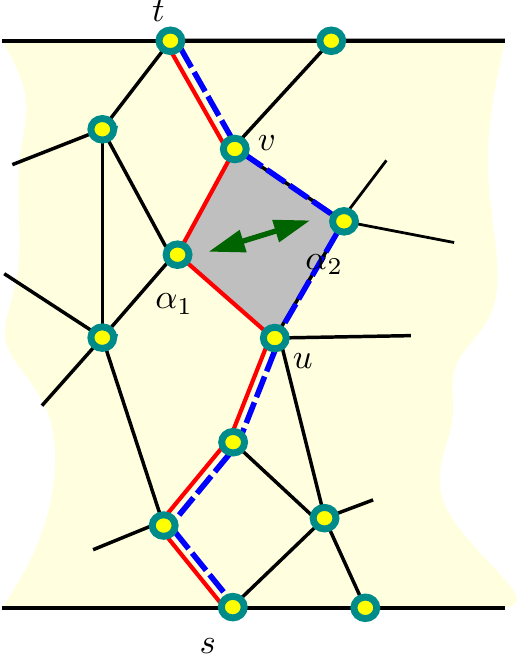} \hfill
    \IncludeGraphics[page=2,scale=0.67]{figs/face_flip} \hfill
    \IncludeGraphics[page=3,scale=0.67]{figs/face_flip} \hfill
    \IncludeGraphics[page=4,scale=0.67]{figs/face_flip}
    \caption{From left to right: face-flip, spike/reverse spike,
       person-move and dog-move.}
    \figlab{F:face-edge-flip}
\end{figure}

\subsubsection{The discrete version}
%\apndlab{discrete:h}

Let $W_1$ be an $(s,t)$-walk and $f$ be a face in an embedded planar
graph $\Graph$.  Assume that $\alpha_1$ is a subwalk of $W_1$ and
$\partial f = \alpha_1 \concat \alpha_2$, where $\alpha_1$ and
$\alpha_2$ are walks that share endpoints $u$ and $v$, such that $u$
is closer to $s$ on $W_1$.  The \emphi{face flip} operation is the
following: The walk $W_2 = W_1[s,u] \cdot \alpha_2 \cdot W_1[v,t]$ is
the result of flipping $W_1$ over $f$.  In this case, $W_1$ and $W_2$
are one face flip operation apart. See \figref{F:face-edge-flip}.

Now, let $W_1$ be an $(s,t)$-walk, and suppose
$W_1 = W'_1\cdot W''_1$.  Also, let $u$ be the common endpoint of
$W'_1$ and $W''_1$, and let $e=(u,v)$ be any edge in $\Graph$.  By
applying a \emphi{spike} to $W_1$ we obtain
$W_2 = W'_1\cdot(u,v)\cdot(v,u)\cdot W''_1$.  Equivalently, we can
obtain $W_1$ from $W_2$ via a \emphi{reverse spike}.  In this case,
$W_1$ and $W_2$ are one spike operation apart.

In general, $W_1$ and $W_2$ are one operation apart if one can
transform one to the other using a single face flip, spike, or reverse
spike.  Chambers and Letscher \cite{cl-hh-09,cl-ehh-10} introduce the
same set of operations with the names: face lengthening, face
shortening, spike and reverse spike.

Let $\CLeft$ and $\CRight$ be two $(s,t)$-walks on the outer face of
$\Graph$.  The sequence of walks $(L=W_0, W_1, \dots, W_m=R)$ is a
$(L, R)$-\emphi{discrete homotopy} if, for $i=1,\ldots, m$, $W_{i}$
and $W_{i-1}$ are one operation apart.  We may use the word homotopy
as a short form of discrete homotopy when it is clear from the
context.
% A homotopy is \emphi{monotonic} (or equivalently it avoids backward
% moves) if $W_{i-1}$ is inside the disc with boundary
% $\CLeft \cup W_i$ for every $1 \leq i \leq m$.
The height of the homotopy is defined to be the length of the longest
walk in its sequence.  The homotopy height between $\CLeft$ and
$\CRight$, is the height of the shortest possible
$(\CLeft,\CRight)$-homotopy.

\begin{definition}
    \deflab{dog:person:move}%
    Let $\curveA = (\pntOnA_0, \pntOnA_1, \dots, \pntOnA_k)$ and
    $\curveB = (\pntOnB_0, \pntOnB_1, \dots, \pntOnB_{k'})$ be walks
    of $\Graph$.  The walk
    \begin{math}
        W_1 = \bigl( \pntOnA_i = w_1,w_2,\dots, w_k = \pntOnB_j \bigr)
    \end{math}
    changes to the walk
    \begin{math}
        W_2 = \bigl( \pntOnA_{i+1}, \pntOnA_i = w_1,w_2,\dots, w_k
        \bigr)
    \end{math}
    after a \emphi{person move}.  Similarly, the walk
    \begin{math}
        W_1 = \bigl( \pntOnA_i = w_1,w_2,\dots, w_k = \pntOnB_j \bigr)
    \end{math}
    changes to the walk
    \begin{math}
        W_2 = \bigl( w_1,w_2,\dots, w_k = \pntOnB_j, \pntOnB_{j+1}
        \bigr)
    \end{math}
    after a \emphi{dog move}.  A \emphi{leash operation} is a dog
    move, a person move, a face flip, a spike or a reverse spike.
\end{definition}

\begin{definition}
    \deflab{leash:sequence}%
    An $(\curveA,\curveB)$-walk is a walk that has one endpoint on
    $\curveA$ and one endpoint on $\curveB$.  A sequence of
    $(\curveA,\curveB)$-walks, $(W_1, W_2, \dots, W_q)$ is called an
    $(\curveA,\curveB)$-\emphi{leash sequence} if %
    \smallskip
    \begin{compactenum}[\qquad(i)]
        \item $W_1$ is a $(\pntOnA_0,\pntOnB_0)$-walk,
        \item $W_q$ is a $(\pntOnA_k, \pntOnB_{k'})$-walk, and
        \item we have that, for $i=1,\ldots, q-1$, $W_i$ changes to
        $W_{i+1}$ by performing a sequence of leash operations
        containing either at most one dog move and no person moves, or
        at most one person move and no dog moves.
    \end{compactenum}
    % \smallskip%
    The height of a leash sequence is the length of its longest walk.
\end{definition}

\begin{definition}
    \deflab{d:f:distance}%
    The \emphi{discrete \Frechet distance} of $\curveA$ and $\curveB$
    is the height of the minimum height $(\curveA,\curveB)$-leash
    sequence\footnote{The discrete \Frechet distance defined here is
       different than the more standard definition, which is usually
       defined over sequences of points.}.
    % The leash sequence $(W_1, W_2, \dots, W_q)$ contains no gap if
    % $W_i$
    % changes to $W_{i+1}$ by exactly one leash operation.
    The \emphi{homotopic discrete \Frechet distance} of $\curveA$ and
    $\curveB$ is the height of the minimum height
    $(\curveA,\curveB)$-leash sequence, where two consecutive walks
    differ by a single leash operation (and this is not required in
    the \emph{discrete \Frechet distance}, where two consecutive walks
    might ``jump'').
\end{definition}
% that contains no gap.

% \end{comment}
%%%%%%%%%%%%%%%%%%%%%%%%%%%%%%%%%%%%%%%%%%%%%%%%%%%%%%%%%%%%%%%% 
\section{Approximating the height -- the discrete case}
\seclab{hh_discrete}
% \label{sec:algorithm_discrete}

In this section, we give an approximation algorithm for finding a
discrete homotopy of minimum height in a topological disk $\Disk$,
whose boundary is defined by two walks $\CLeft$ and $\CRight$ that
share their endpoints $s$ and $t$.  The disk $\Disk$ is a triangulated
edge-weighted planar graph. The ideas developed here are used later
for the continuous case, see \secref{h:h:continuous}.

\subsection{Preliminaries}

We are given an edge-weighted plane graph $\Graph$ all of whose faces
(except possibly the outer face) are triangles. Let
$s,t \in \bd{\Graph}$ and $\CLeft$ and $\CRight$ be two non-crossing
$(s,t)$-walks on $\bd{\Graph}$ in counter-clockwise and clockwise
order, respectively.  We use $\Disk$ to denote the topological disk
enclosed by $\CLeft \concat \CRight$. The vertices of $\Graph$ (inside
or on the boundary of $\Disk$) are also the vertices of $\Disk$.  Our
goal is to find a minimum height homotopy from $\CLeft$ to $\CRight$
of non-crossing walks. Recall that a homotopy is a sequence of walks,
where every two consecutive walks differ by either a triangle, or an
edge (being traversed twice).

\begin{lemma}%
    \lemlab{diam}%
    Let $x$ and $y$ be vertices of $\Graph$ that are at graph distance
    $\rho$.  Then any discrete homotopy between $\CLeft$ and $\CRight$
    has height at least $\rho$.
\end{lemma}

\begin{proof}
    Fix a homotopy of height $\delta$.  This homotopy contains an
    $(s,t)$-walk $\walk$ that passes through $x$, and an $(s,t)$-walk
    $\walkA$ that passes through $y$.  We have, by the triangle
    inequality, that
    \begin{math}
        \rho = \distZ{\Graph}{x}{y} \leq \lenX{\walk[s,x]} +
        \lenX{\walkA[s,y]},
    \end{math}
    and, similarly,
    $\rho \leq \lenX{\walk[x,t]} + \lenX{\walkA[y,t]}$.  Therefore,
    \begin{math}
        \rho \leq (\lenX{\walk} + \lenX{\walkA})/2 \leq \max\pth{
           \lenX{\walk}, \lenX{\walkA}} \leq \delta,
    \end{math}
    as required.%
    \InDCGVer{\qed} %
\end{proof}

\begin{lemma}%
    \lemlab{max:d}%
    Suppose $d_1$ is the maximum distance of a vertex of $\Graph$ from
    $\CLeft$, $d_2$ is the largest edge weight, and let
    $\dLeft=\max\brc{d_1, d_2}$.  Furthermore, let $\Disk$, $\CLeft$,
    and $\CRight$ be defined as above.  Then any discrete homotopy
    between $\CLeft$ and $\CRight$ has height at least $\dLeft$.
\end{lemma}
\begin{proof}
    The height is at least $d_1$ by \lemref{diam}.  On the other hand,
    for every homotopy between $\CLeft$ and $\CRight$, and for every
    edge $\edge$, there exists a walk in the homotopy that passes
    through $\edge$.  Therefore, the height of the homotopy is at
    least $d_2$.%
    \InDCGVer{\qed} %
\end{proof}

\subsection{The algorithm}

\begin{theorem}%
    \thmlab{discrete}%
    Let $\Disk$ be an edge-weighted triangulated topological disk with
    $n$ faces such that its boundary is formed by two walks $\CLeft$
    and $\CRight$ that share endpoints $s$ and $t$.  Then one can
    compute, in $O( n \log n)$ time, a homotopy from $\CLeft$ to
    $\CRight$ of height at most
    $\lenX{\CLeft} + \lenX{\CRight} + O( \dLeft \log n)$, where
    $\dLeft$ is the largest among
    \begin{inparaenum}[(i)]
        \item the maximum distance of a vertex of $\Disk$ from
        $\CLeft$, and
        \item the maximum edge weight.
    \end{inparaenum}
    
    In particular, the generated homotopy has height
    $O\pth{ \HHopt \log n }$, where $\HHopt$ is the minimum homotopy
    height between $\CLeft$ and $\CRight$.
\end{theorem}

\begin{proof}
    We present a recursive algorithm that reduces the problem to
    sub-problems with a smaller number of triangles.  The recursion
    might create instances where the boundary walks $\CLeft$ and
    $\CRight$ are not interior disjoint.  For such instances, it is
    immediate that one can obtain a solution by computing a homotopy
    independently for each maximal disk bounded by
    $\CLeft \concat \CRight$, and composing them to obtain the desired
    homotopy between $\CLeft$ and $\CRight$.  We may therefore focus
    on the case where $\CLeft$ and $\CRight$ are interior disjoint.
	
    Let $b(\delta, \dLeft, n)$ be the maximum possible height of a
    homotopy obtained by our algorithm for any disk of perimeter
    $\delta$ that is composed of $n$ faces and has $\dLeft$ as defined
    in the statement of the theorem.  We prove by induction that
    \begin{math}
        b(\delta,\dLeft,n) \leq \delta +c_0\dLeft\log n
    \end{math}
    for some constant $c_0$, implying the theorem statement. Note that
    the inductive argument implies that $b$ is linear in $\delta$.
    
    The base case $n=0$ is easy. Indeed, if we have two edges $(u,v)$
    and $(v,u)$ consecutive in $\CRight$ (or in $\CLeft$) we can
    retract these two edges. By repeating this we arrive at both
    $\CLeft$ and $\CRight$ being identical, and we are done.  The case
    $n=1$ is handled in a similar fashion. After one face flip, the
    problem reduces to the case $n=0$.  Hence,
    $b(\lenX{\CLeft} + \lenX{\CRight}, \dLeft, 1) \leq \lenX{\CLeft} +
    \lenX{\CRight} + \dLeft$.
    
    For $n > 1$, compute for each vertex of $\Graph$ its shortest path
    to $\CLeft$, and consider the set of edges $\EdgeSet$ used by all
    these shortest paths. Clearly, these shortest paths can be chosen
    so that $\CLeft \cup \EdgeSet$ form a tree. We consider each edge
    of $\CRight$ to be ``thick'' and have two sides (i.e., we think
    about these edges as being corridors -- this is done to guarantee
    that in the recursive scheme, done below, there are exactly two
    subproblems to each instance). If $\EdgeSet$ uses an edge of
    $\CRight$ then it uses the inner copy of this edge, while
    $\CRight$ uses the outer side. Similarly, we consider each
    original vertex of $\CRight$ to be two vertices (one inside and
    the other one on the boundary $\CRight$). The set $\EdgeSet$ would
    use only the inner vertices of $\CRight$, while $\CRight$ would
    use only the outer vertices.  To keep the graph triangulated we
    also arbitrarily triangulate inside each thick edge of $\CRight$
    by adding corridor edges.  Each corridor edge either connects two
    copies of a single vertex (thus has weight zero) or copies of two
    neighbors on $\CRight$ (and so has the same weight as the original
    edge).

    \begin{figure}[t]
        \centerline{%
           \begin{tabular}{ccc}
             \includegraphics{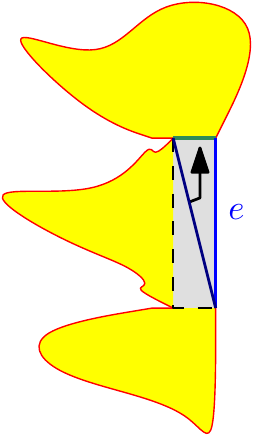}%
             &%
               \qquad
               \qquad
               \IncludeGraphics{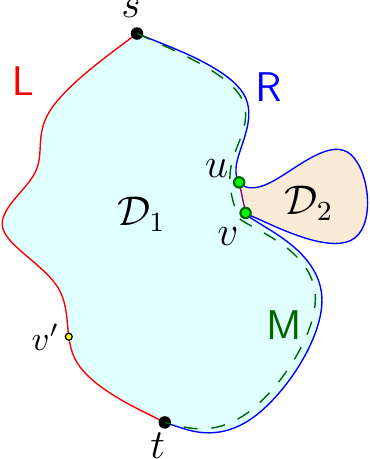}
               \qquad
               \qquad
             &%
               {\IncludeGraphics{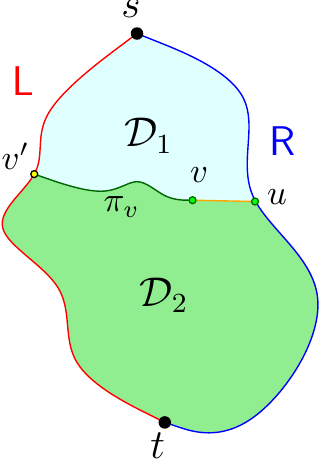}}%
             \\
             (I) 
             &
               (II)
             &
               (III)
           \end{tabular}
        }
        \caption{(I) Corridors edges. (II) Case (A). (III) Case (B). }
        \figlab{corr}
    \end{figure}
    
    Clearly, if we cut $\Disk$ along the edges of $\EdgeSet$, what
    remains is a simple triangulated polygon (it might have ``thin''
    corridors along the edges of $\CRight$). One can find a diagonal
    $uv$ such that each side of the diagonal contains at least
    $\lceil n/3 \rceil$ triangles and at most $\lfloor2n/3\rfloor$
    triangles of the \emph{original} $\Graph$. We emphasize that we
    count only the ``real'' triangles of $\Graph$.  This can be
    achieved as follows: We first assign weight zero to the faces
    within corridors and unit weight to all other faces.  Then we find
    a diagonal $uv$ such that each side contains faces with total
    weight at least $\lceil n/3 \rceil$.  Furthermore, because the
    faces inside corridors have weight zero, we can ensure that if the
    separating edge $uv$ is a corridor edge (i.e., corresponding to an
    edge $\edge$ of $\CRight$) then $u$ and $v$ are copies of the same
    vertex. Indeed, if not, we can change the separating edge so this
    property holds, and the new separating edge separates regions with
    the same weight, see \figref{corr} (I). We use this property in
    the following case analysis\footnote{Note, that the corridors were
       used only in generating this partition, and are an artifact
       that is not necessarily sent to the recursive subproblems. In
       particular, one can describe this partition scheme without
       using the corridors, but it seems somewhat messier and less
       intuitive.}.

    % , see \secref{h:h:continuous}.
    
    \begin{inparaenum}[\smallskip\noindent\bf(A)]
        \item Consider the case\footnote{Strictly speaking this case
           is not possible because of the corridor
           diagonals. Nevertheless, it provides a good warm-up
           exercise for the followup cases which are more involved.}
        that $u$ and $v$ are both vertices of $\CRight$. In this case,
        let $\CRight[u,v]$ be the portion of $\CRight$ in between $u$
        and $v$, and let $\Disk_2$ be the disk having
        $\CRight[u,v] \concat uv$ as its outer boundary. Let $\Disk_1$
        be the disk $\Disk \setminus \Disk_2$. Let
        \begin{math}
            \CMid = \CRight[s,u] \concat uv \concat \CRight[v,t],
        \end{math}
        see \figref{corr} (II).
        
        Clearly, the distance of any vertex of $\Disk_1$ from $\CLeft$
        is at most $\dLeft$. By induction, there is a homotopy of
        height
        \begin{math}
            b( \lenX{\CLeft} + \lenX{\CMid}, \dLeft, \lfloor 2n/3\rfloor
            )
        \end{math}        
        from $\CLeft$ to $\CMid$.  Similarly, the distance of any
        vertex of $\Disk_2$ from $uv$ is at most its distance to
        $\CLeft$. Therefore, by induction, there is a homotopy between
        $uv$ and $\CRight[u,v]$ of height at most
        $b\pth{ \lenX{\CRight[u,v]} + \dLeft, \dLeft, \lfloor
           2n/3\rfloor }$.
        Clearly, we can extend this to a homotopy of $\CMid$ to
        $\CRight$ of height
        $\lenX{\CRight[s,u]} + b\pth{ \lenX{\CRight[u,v]} + \dLeft,
           \dLeft, \lfloor 2n/3\rfloor } + \lenX{\CRight[v,t]}$
        which using induction hypothesis is at most
        $\lenX{\CRight} + \dLeft + c_0\dLeft\log \lfloor 2n/3\rfloor
        \leq \lenX{\CRight} + c_0\dLeft\log n$,
        for sufficiently large $c_0$.
        
        Putting these two homotopies together results in the desired
        homotopy from $\CLeft$ to $\CRight$.
        
        \smallskip
        
        \item Consider the case that $v$ is a vertex of $\EdgeSet$ and
        $u$ is a vertex of $\CRight$. So, $v$ is an inner vertex of
        $\CRight$ (that belongs to $\EdgeSet$) and $u$ is an outer
        vertex of $\CRight$.  Recall that we can assume that $v$ and
        $u$ are inner and outer copies of the same vertex of
        $\CRight$.  Let $\pi_v$ be the shortest path in $\Disk$ from
        $v$ to $\CLeft$, and let $v'$ be its endpoint on $\CLeft$.
                
        Consider the disk $\Disk_1$ having the ``left'' boundary
        $\CLeft_1 = \CLeft[s,v'] \concat \pi_v \concat v u$ and
        $\CRight_1 = \CRight[s, u]$ as its ``right'' boundary, see
        \figref{corr} (III). This disk contains at most
        $\lfloor 2n/3\rfloor$ triangles, and by induction, it has a
        homotopy of height
        \begin{math}
            b(\lenX{\CLeft_1} + \lenX{\CRight_1}, \dLeft, \lfloor
            2n/3\rfloor).
        \end{math}
        To see why we can apply the recursion, observe that $u$ and
        $v$ are copies of the same vertex of $\CRight$.  That is, all
        shortest paths of vertices inside $\Disk_1$ to $\CLeft$ are
        completely inside $\Disk_1$.  As such, the distance of all
        vertices in $\Disk_1$ to $\CLeft_1$ are at most $\dLeft$.
        
        Similarly, the topological disk $\Disk_2$ with the left
        boundary $\CLeft_2 = uv \concat \pi_v \concat \CLeft[v',t]$
        and the right boundary $\CRight_2 = \CRight[u,t]$ has a
        homotopy of height
        $b\bigl(\lenX{\CLeft_2} + \lenX{\CRight_2}, \dLeft, \lfloor
        2n/3\rfloor \bigr)$.

        \smallskip

        We combine these two homotopies as follows.  Let $\CLeft'$ be
        the walk obtained by concatenating $\CLeft_1$ and
        $\CLeft_2$. Note that $\CLeft'$ consists of a copy of $\CLeft$
        and two copies of $\pi_v$.  Clearly, there exists a homotopy
        between $\CLeft$ and $\CLeft'$ of height at most
        $\lenX{\CLeft} + 2\dLeft$, which is obtained by a sequence of
        spike moves along $\pi_v$.  We compose the resulting homotopy
        with the homotopy of $\Disk_1$ that moves $\CLeft_1$ to
        $\CRight_1$, followed by the homotopy of $\Disk_2$ that moves
        $\CLeft_2$ to $\CRight_2$.  The result is a homotopy between
        $\CLeft$ and $\CRight$ of height at most
        \begin{align*}
            \max\pth{%
               \begin{array}{c}
                 \lenX{\CLeft} + 2\dLeft, \\%
                 b\pth{ \MakeBig\! \lenX{\CLeft_1} +
                 \lenX{\CRight_1},
                 \dLeft, \lfloor 2n/3\rfloor}
                 + \lenX{\CLeft_2},
                 \\%
                 \lenX{\CRight_1} + b\pth{\MakeBig\! \lenX{\CLeft_2} +
                 \lenX{\CRight_2}, \dLeft, \lfloor 2n/3\rfloor}
               \end{array}
            }.
        \end{align*}
        If the first number is the maximum, we are done.  Otherwise,
        using the induction hypothesis, the above value is at most
        $\lenX{\CLeft} + \lenX{\CRight} + 2\dLeft + c_0\dLeft\log
        \lfloor 2n/3\rfloor$
        which is at most
        $\lenX{\CLeft}+\lenX{\CRight}+c_0\dLeft\log n$ for
        sufficiently large $c_0$.
        
        \begin{figure}[t]
            \centerline{%
               \begin{tabular}{c%
                 cccc}
                 \IncludeGraphics[width=0.18\linewidth]%
                 {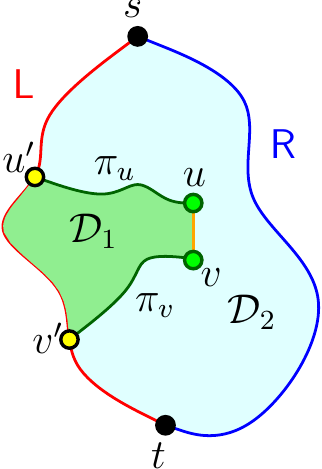}%
                 &%
%                   \quad
                   {\IncludeGraphics[page=2,width=0.17\linewidth]%
                   {figs/general}}
                   % \qquad
                 &%
                   {\IncludeGraphics[page=3,width=0.17\linewidth]%
                   {figs/general}}%
                 &
                   {\IncludeGraphics[page=4,width=0.17\linewidth]% 
                   {figs/general}}%
                 &
                   \IncludeGraphics[page=5,,width=0.17\linewidth]%
                   {figs/general}%
                 \\
                 (I)
                 &
                   (II)%
                 & 
                   (III)%
                 &%
                   (IV)%
                 &%
                   (V)
               \end{tabular}%
            }
            \caption{(I--III) Illustration of case (C) in the proof
               \thmref{discrete}. (IV-V) How the shortest path tree
               get sent to the recursive subproblem.}
            \figlab{C:case}%
        \end{figure}

        \item Here we handle the case that $u$ and $v$ are both
        vertices of $\CLeft \cup \EdgeSet$. Then as before, let $u'$
        and $v'$ be the closest points on $\CLeft$ to $u$ and $v$,
        respectively. Now, let $\pi_u$ (resp. $\pi_v$) be the shortest
        path from $u$ (resp. $v$) to $u'$ (resp.  $v'$). Note that we
        might have $u'=v'$.
        
        Consider the disk $\Disk_1$ having $\CLeft_1 = \CLeft[u', v']$
        as left boundary, and
        $\CRight_1 = \pi_u \concat uv \concat \pi_v$ as right
        boundary, see \figref{C:case}.. This disk contains between
        $n/3$ and $2n/3$ triangles of the original surface. The
        distance of any vertex of $\Disk_1$ to $\CLeft_1$ (when
        restricted to $\Disk_1$) is at most $\dLeft$, and therefore by
        induction, there is a homotopy from $\CLeft_1$ to $\CRight_1$
        of height at most
        $\alpha = b\pth{ \lenX{\CLeft_1} + \lenX{\CRight_1},
           \dLeft,\lfloor 2n/3\rfloor} \leq \lenX{\CLeft[u',v']} +
        3\dLeft +c_0\dLeft\log \lfloor 2n/3\rfloor$.
        This yields a homotopy of height
        $\alpha_1 = \lenX{\CLeft[s,u']} + \alpha +
        \lenX{\CLeft[v',t]}$, from $\CLeft$ to
        \begin{math}
            \CLeft_2 = \CLeft[s,u'] \concat \pi_{u} \concat uv \concat
            \pi_v \concat \CLeft[v',t].
        \end{math}
        It is straightforward to check that
        $\alpha_1 \leq \lenX{L}+3\dLeft +c_0 \dLeft\log \lfloor
        2n/3\rfloor)$.
        
        Next, let $\Disk_2$ be the disk with its left boundary being
        $\CLeft_2$ and its right boundary being $\CRight_2 =
        R$.
        Observe, that as before, the maximum distance of any vertex of
        $\Disk_2$ to $\CLeft_2$ is at most $\dLeft$. As before, by
        induction, there is a homotopy from $\CLeft_2$ to $\CRight_2$
        of height
        $\alpha_2 = b(\lenX{\CLeft_2} + \lenX{\CRight_2}, \dLeft,
        \lfloor 2n/3\rfloor)$.
        Since $\lenX{\CLeft_2} \leq \lenX{\CLeft} + 3d$, we have
        $\alpha_2 \leq b(\lenX{\CLeft}+\lenX{\CRight}+3\dLeft, \dLeft,
        \lfloor 2n/3\rfloor)$.
    \end{inparaenum}
    
    \bigskip%
    \noindent%
    In all cases the length of the homotopy is at most
    \begin{align*}
        \lenX{\CLeft} + \lenX{\CRight}+3\dLeft+c_0\dLeft\log \lfloor
        2n/3\rfloor \leq \lenX{\CLeft} + \lenX{\CRight} +
        c_0\dLeft\log n,
    \end{align*}
    if we choose $c_0$ sufficiently large.  The final guarantee of
    approximation follows as $\dLeft \leq \HHopt$, by \lemref{max:d}.
    
    We can compute the shortest path tree in linear time using the
    algorithm of Henzinger \etal \cite{hkrs-fspap-97}.  The separating
    edge can also be found in linear time using \DFS.  So, the running
    time for a graph with $n$ faces is
    $T(n) = T(n_1) + T(n_2) + O(n)$, where $n_1 + n_2 = n$ and
    $n_1, n_2 \leq (2/3)n$.  It follows that $T(n) = O(n\log n)$. ~%
    \InDCGVer{\qed} %
\end{proof}

\begin{remark}%
    \remlab{no:recompute}%
    \begin{inparaenum}[(A)]
        \item In the algorithm of \thmref{discrete}, it is not
        necessary that we have the shortest paths from $\CLeft$ to all
        the vertices of $\Disk$.  Instead, it is sufficient if we have
        a tree structure that provides paths from any vertex of
        $\Disk$ to $\CLeft$ of distance at most $\dLeft$ in this tree,
        and we send the relevant portions of the tree into the
        recursive subproblems. We will use this property in the
        continuous case, where recomputing the shortest path tree is
        relatively expensive. This is demonstrated in \figref{C:case}
        (IV--V).%

        \item A more careful analysis shows that the height of the
        homotopy generated by \thmref{discrete} is at most
        $\max\pth{\lenX{\CLeft}, \lenX{\CRight}} + O\pth{ \dLeft \log
           n}$.
        
        \item Note, that if
        $\dLeft = O \pth{ \max \pth{ \lenX{\CLeft}, \lenX{\CRight} } /
           \log n }$
        then \thmref{discrete} provides a constant factor
        approximation. This is the situation when $\CLeft$ and
        $\CRight$ are close to each other compared to their relative
        length.
        
        \item Note, that the $O(n\log n)$ running time algorithm
        cannot explicitly output the list of paths in the homotopy.
        Indeed, that list requires $O(n^2)$ space to be stored and so
        $O(n^2)$ time to output.  The output of the algorithm of the
        above lemma is a shortest path tree $T$ together with an
        ordered list of edges.  Each edge $e = (u,v)$ in the list
        represents an $(s,t)$-walk $T[s,u]\cdot (u,v) \cdot T[v,t]$,
        where $T[s,u]$ and $T[v,t]$ are the unique $(s,u)$-path and
        $(v,t)$-path in $T$, respectively.
    \end{inparaenum}
\end{remark}

%%%%%%%%%%%%%%%%%%%%%%%%%%%%%%%%%%%%%%%%%%%%%%%%%%%%%%%%%% 
\section{Approximating the height -- %
   the continuous case}
\seclab{h:h:continuous}

In this section we extend the algorithm from \secref{hh_discrete} to
the continuous case.  The continuous case is somewhat similar to the
solution to the problem of sweeping over the boundary of a convex
polytope in three dimensions from a base point. Since this is
tangential for our main trust, we delegated describing this algorithm
to \apndref{sweep:p}, but the reader might still benefit from reading
it first.

\subsection{Preliminaries}
\seclab{h:prelim}

We are given a piecewise linear triangulated topological disk,
$\Disk$, with $n$ triangles, and we consider the underlying metric to
be the geodesic distance on this surface\footnote{Formally, for two
   points $\pnt, \pntA \in \Disk$, their \emph{geodesic distance} is
   the length of the shortest path inside $\Disk$ connecting $\pnt$
   with $\pntA$.}.  The boundary of $\Disk$ is composed of two paths
$\CLeft$ and $\CRight$ with shared endpoints $s$ and $t$, and the task
at hand is to compute a morphing from $\CLeft$ to $\CRight$ that
minimizes the distance traversed by each point of $\CLeft$ during this
motion. See \secref{homotopy:continuous} for the formal definition.

% Observe that the distance of any point $x$ in $\Disk$ from $\CLeft$
% and $\CRight$ is not longer than the homotopy height as there is a
% $(s,t)$-path that contains $x$. 
Here, we build a homotopy of height at most
$\lenX{\CLeft} + \lenX{\CRight} + O(d\log n)$, where $d$ is the
maximum distance of any point in $\Disk$ from either $\CLeft$ or
$\CRight$.  We use the following observations (see
\secref{s:p:polytope} for details):
\begin{compactenum}[(A)]
    \item The shortest path from a vertex to the whole surface can be
    computed in $O(n^2\log n)$ time.
    \item The shortest path from a set of $O(n)$ edges to the whole
    surface can be computed in $O(n^3\log n)$ time.
    \item A shortest path (i.e., a geodesic) intersects a face along a
    segment and it locally looks like a segment if the adjacent faces
    are rotated to be coplanar.
\end{compactenum}
See \secref{s:p:polytope} for more details.

\subsection{Homotopy height if edges are short}
\seclab{short}

Similar to the discrete case, $d_1$ is the maximum distance for any
point of $\Disk$ from $\CLeft$, $d_2$ is the maximum length of any
edge, and $\dLeft = \max(d_1, d_2)$.  Here, we assume $d_2 \leq 2d_1$.
In this case we can obtain the desired approximation algorithm via an
argument that is similar to the one used in the discrete case.

As in the discrete case, let $\EdgeSet$ be the union of all the
shortest paths from the vertices of $\Disk$ to $\CLeft$ (as before, we
treat the edges and vertices of $\CRight$ as having infinitesimal
thickness). For a vertex $v$ of $\Disk$, its shortest path $\pi_v$ is
a polygonal path that crosses between faces (usually) in the middle of
edges (it might also go to a vertex, merge with some other shortest
paths and then follow a common shortest path back to $\CLeft$). In
particular, each such shortest path might intersect a face of $\Disk$
along a single segment. Thus, the polygon resulting from cutting
$\Disk$ along $\EdgeSet$, call it $P$, is a polygon that has
complexity $O(n^2)$. A face of $P$ is a hexagon, a pentagon, a
quadrilateral, or a triangle.  However, each such face has at most
three edges that are portions of the edges of $\Disk$.  The degree of
a face is $i$ if it has $i$ edges that are portions of the edges of
$\Disk$.  Observe that, each triangle of $\Disk$ is now decomposed
into a set of faces.  Obviously, each triangle of $\Disk$ contains at
most one face of degree $3$ in $P$.  Overall, there are $O(n)$ faces
of degree $3$ in $P$.

Now consider $C^*$, the dual of the graph that is inside the polygon
(ignore the edges on the boundary).  More precisely, $C^*$ has a
vertex corresponding to each face inside the polygon $P$, let $n_p$ be
number of vertices of $C^*$.  Two vertices of $C^*$ are adjacent if
and only if their corresponding faces share a portion of an edge of
$\Disk$ (this shared edge is a diagonal of $P$).  Note that because
$P$ is simply connected $C^*$ is a tree.  Since the maximum degree of
the tree $C^*$ is $3$, there is an edge that is a good separator
(i.e., a separator that has at most $2/3$ of the faces on one side)%
\footnote{The existence of such a tree edge separator is folklore --
   its proof is provided by Lewis \etal \cite{lsh-mbrcf-65}.}.  Since
the length of the edge is at most $2d_1$ it can be used in a similar
fashion as the proof of \thmref{discrete}.  However, in the recursion
of the continuous case we avoid recomputing the shortest paths (i.e.,
we use the old shortest paths and distances computed in the original
disk), see \remref{no:recompute}.  So, we compute the shortest paths
once in the beginning in $O(n^3 \log n)$ time.  Then in each step we
can find the separator in $O\pth{n^2}$ time. Namely, the total time
spent on computing the separators is
$T(n_p) = T(n_1) + T(n_2) + O(n^2)$, where $n_1 + n_2 = n_p$ and
$n_1, n_2 \leq (2/3) (n_1 + n_2)$; since $n_p=O(n^2)$,
$T(n) = O(n^2\log n)$. As such, the total running time is dominated by
the computation of the shortest paths.  The output is a list of
$O(n^2)$ paths each of complexity $O(n)$, and so it can be explicitly
presented in $O(n^3)$ time and space.  Note that, we need a continuous
deformation between any two consecutive paths in the list, which can
be implicitly presented by a collection of functions in linear time
and space (this is similar to what we describe below in the beginning
of \secref{long:edges}).

The proof of \thmref{discrete} then goes through literally in this
case.  Since all the edges have length at most $2d_1$, by assumption,
we obtain the following.

\begin{lemma}%
    \lemlab{continuous}%
    Let $\Disk$ be a topological disk with $n$ faces where every face
    is a triangle (here, the distance between any two points on the
    triangle is their Euclidean distance).  Furthermore, the boundary
    of $\Disk$ is formed by two walks $\CLeft$ and $\CRight$ (that
    share two endpoints $s,t$).  Let $d_1$ be the maximum distance of
    any point of $\Disk$ from $\CLeft$. Finally, assume that all edges
    of $\Disk$ have length at most $2d_1$.  Then one can compute, in
    $O(n^3 \log n)$ time, a continuous homotopy from $\CLeft$ to
    $\CRight$ of height at most   
    \begin{math}
        \lenX{\CLeft} + \lenX{\CRight} + O\pth{ d_1 \log n}.
    \end{math}
\end{lemma}

\subsection{Homotopy height if there are long edges}
\seclab{long:edges}

%%%%%%%%%%%%%%%%%%%%%%%%%%%%%%%%%%%%%%%%%%%%%%%%%%%%%%%%%%% 
\subsubsection{Breaking the disk into strips, pockets and %
   chunks}
\seclab{decomposition}

% \paragraph{Algorithm.}

For any two points in $\Disk$ consider a shortest path $\pi$
connecting them.  The signature of $\pi$ is the ordered sequence of
edges (crossed or used) and vertices used by $\pi$, see
\secref{prelims}.  For a point $\pnt \in \CRight$, let $\signXL{\pnt}$
denote the signature of the shortest path from $\pnt$ to $\CLeft$. The
signature $\signXL{\pnt}$ is well defined in $\CRight$ except for a
finite set of \emphi{medial} points, where there are two (or more)
distinct shortest paths from $\CLeft$ to $\pnt$. In particular, let
$\PSetX{\CRight}$ be the set of all shortest paths from any medial
point on $\CRight$ to $\CLeft$.  Observe that, the medial points are
the only points (on $\CRight$) where the signature of the shortest
path from $\CRight$ to $\CLeft$ changes in any non-degenerate
triangulation.

Cutting $\Disk$ along the paths of $\PSetX{\CRight}$ breaks $\Disk$
into {corridors}.  If the intersection of a corridor with $\CRight$ is
a point (resp.~segment) then it is a \emphi{delta}
(resp.~\emphi{strip}), see \figref{corr:ch}~(I).  In a {strip}
$\strip$, all the shortest paths to $\CLeft$ from the points in the
interior of the segment $\strip \cap \CRight$ have the same signature.
Intuitively, strips have a natural way to morph from one side to the
other. We further break each delta into chunks and pockets, as
follows.

\begin{figure}[t]
    \centerline{%
       \begin{tabular}{cc}%
         {\IncludeGraphics{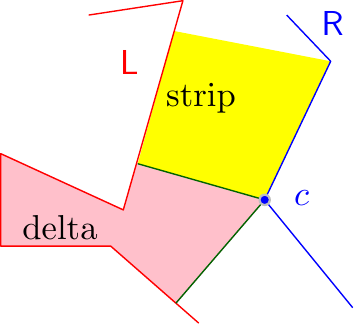}}
         \qquad
         \qquad
         &%
           \qquad
           \qquad
           \IncludeGraphics{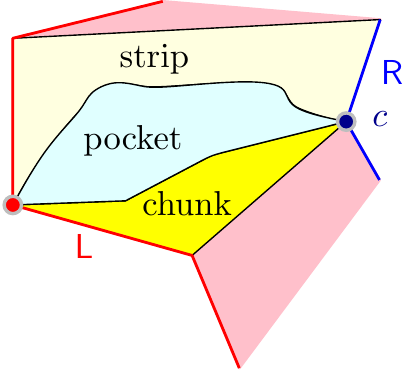}%
         \\
         (I) & (II)
       \end{tabular}%
    }%
    \caption{}
    \figlab{corr:ch}
\end{figure}

Consider a delta $\strip$ with an apex $c$ (i.e., the point of
$\CRight$ on the boundary of $\strip$). For a point
$x \in \CLeft \cap \strip$, its signature (in relation to $\strip$),
is the signature of the shortest path from $x$ to $c$ (restricted to
lie inside $\strip$). Again, we partition $\CLeft \cap \strip$ into
maximum intervals that have the same signature, and let $\PntSet$ be
the set of endpoints of these intervals. For each point
$\pnt \in \PntSet$, consider all the different shortest paths from $c$
to $\pnt$ inside the delta $\strip$, and cut $\strip$ along these
paths. This breaks $\strip$ into regions.  If a newly created region
has a single intersection point with both $\CLeft$ and $\CRight$, then
it is a \emphi{pocket}, otherwise, it is a \emphi{chunk}.  Clearly,
this process decomposes $\strip$ into pockets and chunks.  See
\figref{corr:ch}~(II).

% \medskip

Applying the above partition scheme to all the deltas results in a
decomposition of $\Disk$ into strips, chunks and pockets.

\paragraph{Analysis.}
%\seclab{decomposition:analysis}

Recall, that $d_1$ is the maximum distance of any point of $\Disk$ to
$\CLeft$, and let $d_3$ be the maximum distance of any point of
$\Disk$ to $\CRight$.  Also, let 
\begin{align}
    \dBoth = \max(d_1, d_3).
    \eqlab{d:both}
\end{align}
Now, consider a chunk $\chunk$. Its intersection with $\CLeft$ is a
segment, and its intersection with $\CRight$ is a point (i.e., the
apex $c$ of the delta).

\begin{lemma}
    \lemlab{no:vertex}%
    A strip $\strip$ cannot have any vertex of $\Disk$ in its interior.
\end{lemma}

\begin{proof}
    Let $\edge_\CLeft = \CLeft \cap \strip$ and
    $\edge^{}_\CRight = \CRight \cap \strip$ be two edges bounding a
    strip.  For two points $\pnt, \pnt'$ in the interior of
    $\edge^{}_\CRight$, consider their corresponding shortest paths
    $\PathA$ and $\PathA'$ to $\CLeft$. By definition, these two paths
    have the same signature $\signXL{\PathA} = \signXL{\PathA'}$.  If
    not, then by a limit argument, there must be a point
    $\pnt'' \in \edge^{}_\CRight$ in between these two points, which
    has two different shortest path with different signature arriving
    to it; that is, $\pnt''$ is a medial point, implying that
    $\edge^{}_\CRight$ is broken into (at least) two edges, and it can
    not be the right side of a strip.

    Now, for $i=1,\ldots, m$, let $\edge^{}_i$ be the $i$\th edge of
    $\Disk$ that intersect $\PathA$, as we move from $\CLeft$ to
    $\CRight$ along $\PathA$.  Observe that for any $i$, the edges
    $\edge_i,\edge_{i+1}$ belong to some triangle $\triangle_i$ of
    $\Disk$, which $\PathA$ and $\PathA'$ goes through. In particular,
    being shortest paths, $\PathA$ and $\PathA'$ each intersect
    $\triangle_i$ along a segment. In particular, let $B_i$ be the
    region of $\triangle_i$ bounded by $\PathA$ and $\PathA'$. The
    region $B_i$ does not contain any vertex of $\Disk$ in its
    interior, and it thus follows that the region $\strip$ enclosed
    between $\PathA$ and $\PathA'$ (i.e., $\bigcup_i B_i$) does not
    contain any vertex of $\Disk$.  Now, applying this argument to a
    sequence of points $(\pnt,\pnt')$ that converge to the endpoint of
    $\edge_\CRight$, implies the claim.%
    \InDCGVer{\qed} %
\end{proof}

\begin{remark:unnumbered}
    \lemref{no:vertex} testifies that no vertex of $\Disk$ can be
    interior to a strip. However, strangely enough, a strip might be
    pinched together by some middle vertices. To see that, visualize a
    terrain with saddle points (i.e., passes high in the mountains),
    and the strip is made out of two triangle like shapes (with
    $\edge^{}_\CLeft$ and $\edge^{}_\CRight$ as their respective
    bases), connected by the unique path between the two extreme
    saddle points\footnote{Thus, a strip might look like a dissected
       butterfly. Sad indeed.}.
\end{remark:unnumbered}

The somewhat more challenging case to handle is that of pockets. A pocket
is a topological disk such that its intersections with $\CLeft$ and
$\CRight$ are both single points, and the two boundary paths between
these intersections are of length at most $2\dBoth$. The overall
perimeter of a pocket is of length at most $4 \dBoth$, see
\figref{pocket} (I).  Pockets are handled by using the recursive
scheme developed for the discrete case.

\begin{figure}[t]
    \centerline{%
       \begin{tabular}{ccc}      
         \includegraphics{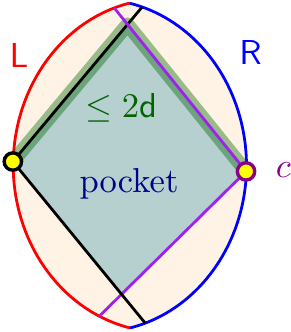}%
         \qquad
         &%
           \qquad{\includegraphics{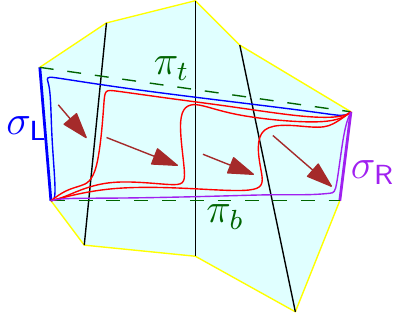}}%
         \\%
         (I) & (II)
       \end{tabular}%
    }
    \captionof{figure}{}%
    \figlab{pocket}%
\end{figure}

\subsubsection{The algorithm in detail}

We use the algorithm of \secref{decomposition} to break the given disk
$\Disk$ into strips, chunks and pockets (notice, that we assume
nothing on the length of the edges).  Next, order the resulting
regions according to their order along $\CLeft$, and transform each
one of them at time, such that starting with $\CLeft$ we end up with
$\CRight$. In each such chunk or strip, the homotopy has height
(roughly) proportional to its perimeter, while for a pocket the
situation is more involved.

\medskip%
\SaveIndent%
\begin{compactenum}[\quad(A)]
    \RestoreIndent
    \item \textbf{Morphing a chunk/strip $S$}: Let
    $\BPath_\CLeft = \CLeft \cap S$ and
    $\BPath_\CRight = \CRight \cap S$, and let $\pi_t$ and $\pi_b$ be
    the top and bottom paths forming the two sides of $S$. There is a
    natural homotopy from $\pi_t \concat \BPath_\CLeft$ to
    $\BPath_\CRight \concat \pi_b$.
    
    The strip/chunk $S$ has no vertex of $\Disk$ in its interior, and
    therefore it is formed by taking planar quadrilaterals and gluing
    them together along common edges.  Observe that by the triangle
    inequality, all such edges of any of these quadrilaterals are of
    length at most
    $\max\pth{\lenX{\BPath_\CLeft}, \lenX{\BPath_\CRight}} + 4\dBoth$.
    It is now easy to check that we can collapse each such
    quadrilateral in turn to obtain the required homotopy.  Since each
    of $\pi_t$ and $\pi_b$ is composed of two shortest paths
    (\figref{pocket} demonstrates why such a path can potentially be
    made of two shortest paths), there is a linear number of such
    quadrilaterals.  See \figref{pocket} (II) for an example.
    
    \medskip

    \item \textbf{Morphing a pocket}: We apply the algorithm of
    \lemref{continuous} recursively to a pocket.

\end{compactenum}
\medskip%
Specifically, the above decomposes $\Disk$ into $m$
chunk/strips/pockets: $\Disk_1, \ldots, \Disk_m$ ordered by their
intersection with $\CLeft$. Each such disks $\Disk_i$ has a left
(resp. right) subcurve $\CLeft_i = \CLeft \cap \Disk_i$,
(resp. $\CRight_i = \CLeft \cap \Disk_i$), and similarly, it has a top
curve $\CTop_i = \Disk_{i-1} \cap \Disk_i$ and a bottom curve
$\CBot_i = \Disk_{i} \cap \Disk_{i+1}$, for $i=1,\ldots, m$.  In the
end of the $i$\th iteration, of this morphing process, the current
curve is going to be
\begin{align*}
    \CMid_i = \CRight_1 \concat \ldots \concat \CRight_i \concat
    \CBot_i \concat \CLeft_{i+1} \concat \cdots \concat \CLeft_m.
\end{align*}
Specifically, at the $i$\th iteration, the algorithm morph
$\CMid_{i-1}$ to $\CMid_i$, as described above (depending on what kind
of region it is). In particular, initially $\CMid_0 = \CLeft$ and in
the end $\CMid_m = \CRight$.  As such, this results in the desired
homotopy.

\subsubsection{Analysis}

\paragraph{Why can we apply \lemref{continuous} to a pocket.} 
A pocket has perimeter at most $4\dBoth$, and there is a point on its
boundary, such that the distance of any point in it to this base point
is at most $2\dBoth$. Indeed, the boundary of $\dBoth$ in the worst
case, is made out of four shortest paths in the original disk, and as
such, its total length is at most $4\dBoth$, see
\Eqref{d:both}. Furthermore, the distance of any point in the pocket
to the apex on $\CRight$, is at most $2\dBoth$.

Now, by the triangle inequality, we have that if in a topological disk
$\Disk$ all the points of $\Disk$ are in distance at most $2\dBoth$
from some point $c$, then the longest edge in $\Disk$ has length at
most $4\dBoth$.  Therefore, all the edges inside a pocket cannot be
longer than $4\dBoth$.

\paragraph{Running time.}
The shortest paths from $\CRight$ to $\CLeft$ can be computed in
$O(n^3\log n)$ time.  The shortest paths inside a delta to its apex
can be computed in $O(n^2 \log n)$ time.  Since there is a linear
number of deltas, the total running time for building the strips is
$O(n^3\log n)$.

\begin{lemma}%
    \lemlab{medials}%
    The number of paths in $\PSetX{\CRight}$ is
    $O\pth{ \bigl. \cardin{\Vertices{\Disk}}}$, where
    $\Vertices{\Disk}$ is the set of vertices of $\Disk$.

    In particular, the total number of parts (i.e., strips, chunks,
    and pockets) generated by the above decomposition is
    $O\pth{\bigl.\cardin{\Vertices{\Disk}}}$.
\end{lemma}

\begin{proof}
    Let $\brc{\BPath_1, \BPath_2, \dots, \BPath_k}$ be the paths in
    $\PSetX{\CRight}$ sorted by the order of their endpoints along
    $\CRight$. Observe that these paths are geodesics and so one can
    assume that they are interior disjoint, or share a suffix (or a
    prefix). Now, if $l_i \in \CLeft$ and $r_i \in \CRight$ are the
    endpoints of $\BPath_i$, for $i=1,\ldots, k$, then these endpoints
    are sorted along their respective curves. In particular, let
    $\Disk_i$ be the disk with boundary
    \begin{math}
        \CLeft[s, l_i] \concat \BPath_{i+1} \concat \CRight[s,r_i].
    \end{math}
    We have that
    \begin{math}
        \Disk_1 \subseteq \Disk_2 \subseteq \cdots \subseteq \Disk_k.
    \end{math}
    The signatures of $\BPath_i$ and $\BPath_{i+2}$ must be different
    as otherwise they would be consecutive. Furthermore, because of
    the inclusion property, if an edge or a vertex of $\Disk$
    intersects $\BPath_i$ but does not intersect $\BPath_{i+1}$ then
    it cannot intersect any later path. Therefore, every other path in
    $\PSetX{\CRight}$ can be charged to vertices or edges that are
    added or removed from the signature of the respective path.  Since
    an edge or a vertex can be added at most once, and deleted at most
    once, this implies the desired bound on the number of paths.

    The second claim follows readily by the above.%
    \InDCGVer{\qed} %
\end{proof}

The following bounds the quality of the morphing for a pocket or a
chunk.

\begin{lemma}%
    \lemlab{strip}%
    Consider a strip or a chunk $S$ generated by the above partition
    of $\Disk$. Let $\BPath_\CLeft = \CLeft \cap S$ and
    $\BPath_\CRight = \CRight \cap S$. Let $\pi_t$ and $\pi_b$ be the
    top and bottom paths forming the two sides of $S$ that do not lie
    on $\CRight$ or $\CLeft$.
    \begin{compactenum}[\quad(A)]
        \item We have $\lenX{\pi_b} \leq 2\dBoth$ and
        $\lenX{\pi_t} \leq 2\dBoth$.
        \item If $\,\lenX{\BPath_\CLeft} > 0 $ or
        $\lenX{\BPath_\CRight} >0 $ then there is no vertex of $\Disk$
        in the interior of $S$.
        \item If $\,\lenX{\BPath_\CLeft} > 0 $ or
        $\lenX{\BPath_\CRight} >0 $ then there is a homotopy from
        $\pi_t \concat \BPath_\CLeft$ to
        $\BPath_\CRight \concat \pi_b$ of height
        $\max\pth{\lenX{\BPath_\CLeft}, \lenX{\BPath_\CRight}} +
        4\dBoth$. We can compute such a homotopy in linear time.
    \end{compactenum}
\end{lemma}
\begin{proof}
    (A) If the strip was generated by the first stage of partitioning
    then the claim is immediate.
    
    Otherwise, consider a delta $C$ with an apex $c$. For any point
    $x \in \CLeft \cap C$ we claim that there is a path of length at
    most $2\dBoth$ to $c$. Indeed, consider the shortest path $\pi_x$
    from $x$ to $\CRight$ in $\Disk$. If this path goes to $c$ the
    claim holds immediately. Otherwise, the shortest path (that has
    length at most $\dBoth$) must cross either the top or bottom
    shortest path forming the boundary of $C$ that are emanating from
    $c$. We can now modify $\pi_x$, so that after its intersection
    point with this shortest path, it follows it back to $c$. Clearly,
    the resulting path has length at most $2\dBoth$ and lies inside
    the resulting chunk.
    
    (B)
    Follows readily from the argument of \lemref{no:vertex}.

% XXX Indeed, the boundary paths $\pi_t$ and $\pi_b$ have the same
%     signature (formally, they are the limit of paths with the same
%     signature).  Since $\Disk$ is non-degenerate, if there was any
%     vertex in the middle, then the path on one side of the vertex, and
%     the path on the other side of the vertex cannot possibly have the
%     same signature.
    
    (C) Immediate from the algorithm description.%
    \InDCGVer{\qed} %
\end{proof}

\subsection{The result}

\begin{theorem}%
    \thmlab{hh}%
    Suppose that we are given a triangulated piecewise linear surface
    with the topology of a disk, such that its boundary is formed by
    two walks $\CLeft$ and $\CRight$. Then there is a continuous
    homotopy from $\CLeft$ to $\CRight$ of height at most
    $\lenX{\CLeft} + \lenX{\CRight} + O\pth{ \dBoth \log n}$, where
    $\dBoth$ is the maximum geodesic distance of any point of $\Disk$
    from either $\CLeft$ or $\CRight$. This homotopy can be computed
    in $O\pth{ n^3 \log n }$ time, and it is a $O( \log n)$
    approximation to the optimal minimum height homotopy.
\end{theorem}
\begin{proof}
    The algorithm is described above. The quality of approximation
    (i.e., $O( \log n)$) follows from plugging in the above into the
    analysis of \thmref{discrete}. Indeed, the intermediate curves
    $\CMid_0, \ldots, \CMid_m$ have length at most
    $\lenX{\CLeft} + \lenX{\CRight} + 2 \dBoth$.  The intermediate
    morphing of a strip or a chunk might result in a curve of length
    $\lenX{\CLeft} + \lenX{\CRight} + 4 \dBoth$, as can be verified
    easily.  As such, any further expansion in the length needed is a
    result of the recursive morphing of a pocket, thus accounting for the
    additional $O\pth{ \dBoth \log n}$ term.
    
    Note, that $\max(\dBoth/2, \CLeft, \CRight)$ is a lower bound on
    the height of the optimal homotopy.%
    \InDCGVer{\qed} %
\end{proof}

% ------------------------------------------------------------------
% ------------------------------------------------------------------

\section{Approximating the homotopic \Frechet %
   distance}
\seclab{HFD}

In this section, fix $\Disk$ to be a triangulated topological disk
with $n$ faces.  Let the boundary of $\Disk$ be composed of $\CTop$,
$\CRight$, $\CBot$, and $\CLeft$, four internally disjoint walks
appearing in clockwise order along the boundary.  Also, let
$\vTopL = \CLeft \cap \CTop$, $\vBotL = \CLeft \cap \CBot$,
$\vTopR = \CRight \cap \CTop$, and $\vBotR = \CRight \cap \CBot$%
\footnotemark.  See \figref{tall}.  \footnotetext{We use the same
   notation to argue about the discrete and continuous problems.} %

\begin{figure}[t]
    \centerline{\IncludeGraphics[page=2,scale=0.9]{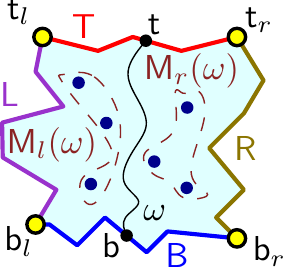}}%
    \caption{}
    \figlab{tall}
\end{figure}%

\subsection{Approximating the regular \Frechet distance}

\subsubsection{The continuous case}

Let $\distFr{\CTop}{\CBot}$ (resp. $\distFrH{\CTop}{\CBot}$) be the
regular (resp. homotopic) \Frechet distance between $\CTop$ and
$\CBot$.  Clearly,
$\distFr{\CTop}{\CBot} \leq \distFrH{\CTop}{\CBot}$.  The following
lemma implies that the \Frechet distance can be approximated within a
constant factor.

\begin{lemma}%
    \lemlab{FD:c}%
    Let $\Disk$, $n$, $\CLeft$, $\CTop$, $\CRight$, and $\CBot$ be as
    above.  Then, for the continuous case, one can compute, in
    $O\pth{n^3 \log n}$ time, reparametrizations of $\CTop$ and
    $\CBot$ of width at most $2\delta$, where
    $\delta = \distFr{\CTop}{\CBot}$.
\end{lemma}
\begin{proof}
    In the following, consider $\Disk$ to be the region bounded by
    these four curves $\CLeft$, $\CTop$, $\CRight$, and $\CBot$.  We
    decompose $\Disk$ into strips, chunks and pockets using the
    algorithm of \secref{decomposition}. Let $\PathSet$ be the set of
    shortest paths from all points of $\CTop$ to the curve $\CBot$.
    As in the algorithm of \secref{decomposition}, let $\PSetX{\CTop}$
    be the set of all shortest paths from medial points on $\CTop$ to
    $\CBot$.  Arguing as in \lemref{medials}, we have that the set
    $\PSetX{\CTop}$ is composed of a linear number of paths.  The
    paths in $\PSetX{\CTop}$ do not cross and so partition $\Disk$
    into a set of regions.  Each region is bounded by a portion of
    $\CTop$, a portion of $\CBot$ and two paths in $\PSetX{\CTop}$.  A
    region is a \emphi{delta} if the two paths of $\PSetX{\CTop}$ in
    its boundary share a single endpoint (on $\CTop$), it is a
    \emphi{pocket} if they share two endpoints (one on $\CTop$ and one
    on $\CBot$), and it is \emphi{strip} if they share no endpoints.
    
    Obviously, the (endpoints of the) paths in $\PathSet$ cover all of
    the vertices of $\CTop$.  The paths in $\PathSet$ also cover all
    of $\CBot$ except for the bases of deltas.  Now, for each delta we
    compute the set of all shortest paths from the vertices of its
    base to its apex inside the delta.  Let $\PSetX{\CBot}$ be the set
    of all such paths in all deltas.  Clearly, the union of
    $\PSetX{\CBot}$ and $\PSetX{\CTop}$ is a set of non-crossing paths
    whose endpoints cover all the vertices of $\CTop$ and $\CBot$.
    
    The shortest path from any point of $\CTop$ to $\CBot$ is at most
    $\delta$.  So, all paths in $\PathSet$ have length at most
    $\delta$.  Similarly, the shortest path from a point of $\CBot$ to
    $\CTop$ is at most $\delta$.  Now, consider a delta $C$ with apex
    $c$.  Let $b$ be a point on the base of $C$ (and so on $\CBot$).
    The shortest path $\pi_b$ from $b$ to $\CTop$ has length at most
    $\delta$.  Let $x$ be the first point that $\pi_b$ intersects a
    boundary path of $C$, $\pi_C$.  Now,
    $\pi_b[b,x] \concat \pi_C[x,c]$ has length at most $2\delta$ and
    it is inside $C$.  We conclude that all paths in $\PSetX{\CBot}$
    have length at most $2\delta$.
    
    The paths in $\PSetX{\CBot} \cup \PSetX{\CTop}$ decompose $\Disk$
    into strips and corridors. The left and right portions of a strip
    is of length at most $2\delta$, and its top and bottom sides have
    as such \Frechet distance at most $2\delta$ from each
    other. Similarly, the leash can jump over a pocket from the left
    leash to the right leash. Doing this to all corridors and pockets,
    results in reparametrizations of $\CLeft$ and $\CRight$ such that
    their maximum length of a leash for these reparametrizations are
    at most $2\delta$.  This implies that the \Frechet distance is at
    most $2\delta$, and we have an explicit reparametrization that
    realizes this distance.
    
    As for the running time, in $O(n^3 \log n)$ time, one can compute
    all shortest paths from $\CTop$ to the whole surface.  Then one
    can, in $O(n^2\log n)$ time, compute the shortest paths inside
    each of the linear number of deltas.  It follows that the total
    running time is $O(n^3\log n)$.%
    \InDCGVer{\qed} %
\end{proof}

\subsubsection{The discrete case}

We can use a similar idea to the decomposition into deltas, pockets
and strips as done in the proof of \lemref{FD:c}.

\begin{lemma}%
    \lemlab{FD:d}%
    Let $\Disk$ be a triangulated topological disk with $n$ faces, and
    $\CTop$ and $\CBot$ be two internally disjoint walks on the
    boundary of $\Disk$.  Then, for the discrete case, one can
    compute, in $O\pth{ n }$ time, reparametrizations of $\CTop$ and
    $\CBot$ that approximate the discrete \Frechet distance between
    $\CTop$ and $\CBot$. The computed reparametrizations have width at
    most $3\delta$, where $\delta$ is the maximum of the \Frechet
    distance between $\CTop$ and $\CBot$, and the maximum length of an
    edge on
\end{lemma}
\begin{proof}
    First, compute the set of shortest paths,
    $\PSetX{\CTop} = \brc{\PathA_1, \PathA_2, \cdots, \PathA_k}$, from
    vertices of $\CTop$ to the path $\CBot$. To this end, we
    (conceptually) collapse all the vertices of $\CBot$ into a single
    vertex, and compute the shortest path from this meta vertex to all
    the vertices in $\CTop$. Let $\Tree$ be the resulting shortest
    path tree.
    
    Next, for $i=1, \ldots, k-1$, let $\edge_i = \vTop_i\vTop_{i+1}$
    be the $i$\th edge of $\CTop$, and let let $\PathA_i$ be the
    shortest path from $\vTop_i$ to $\CBot$, with $\vBot_i$ being its
    endpoint on $\CBot$, and consider the region $\Disk_i$ bounded by
    the curve
    \begin{math}
        \PathA_i \concat \edge_i \concat \PathA_{i+1} \concat
        \SubCurve{\CBot}{\vBot_{i+1}}{\vBot_i}.
    \end{math}
    Now, compute the shortest path tree $\Tree_{i}$ inside $\Disk_i$,
    from the two vertices of $\edge_i$ to all the other vertices of
    $\Disk_i$. For each internal vertex $\vrt$ of
    $\SubCurve{\CBot}{\vBot_{i+1}}{\vBot_i}$, the shortest path to
    either $\vTop_i$ or $\vTop_{i+1}$ inside $\Disk_i$ can be
    retrieved from $\Tree_i$.  Let $\PSetX{i}$ be the set of all such
    shortest paths for internal vertices of
    $\SubCurve{\CBot}{\vBot_{i+1}}{\vBot_i}$, and let
    $\PathSet = \PSetX{\CTop} \cup \bigcup_i \PSetX{i}$.

    As for the length of the paths in $\PathSet$, observe that the
    shortest path $\PathB$, in $\Disk$, from such a vertex $\vrt$ to
    $\CTop$ has length at most $\delta$. If $\PathB$ wanders outside
    $\Disk_i$ then one can modify it to lie in $\Disk_i$.
    Specifically, if this path intersect, say, $\PathA_i$ then we can
    modify it into a path from $\vrt$ to $\vTop_i$, and the modified
    path has length
    \begin{math}
        \leq \lenX{\PathB} + \lenX{\PathA_i} \leq 2\delta.
    \end{math}

    Now, every edge of $\CTop$ or $\CBot$ must be used by a valid
    leash sequence, see \defref{leash:sequence}. As such, the height
    of any leash sequence is at least the length of the longest such
    edge.  Note, that two consecutive paths in $\PathSet$ might be
    either share an endpoint or adjacent, in either the top or bottom
    curve.  As such, the set $\PathSet$ can be turned into a valid
    leash sequence by adding at most two moves, in the worst case both
    a person and a dog move, between two such consecutive paths. Let
    $\PathSet'$ denote the resulting leash sequence. Now, every path
    in $\PathSet'$ has length at most $2\delta + \delta$, as the
    modified added paths are longer by at most the length of a single
    edge of $\CTop$ or $\CBot$, Thus, the leash sequence $\PathSet'$
    has height at most $3\delta$.
    
    Using the algorithm of Henzinger \etal \cite{hkrs-fspap-97} to
    compute the shortest paths from $\CBot$ takes linear time.  Since
    all the regions are disjoint, and every edge appears on the
    boundary of at most two regions, we can compute all the shortest
    paths inside all these regions to $\CTop$ in $O(n)$ time overall
    (this step requires careful implementation to achieve this running
    time).%
    \InDCGVer{\qed} %
\end{proof}

\begin{remark} %
    \remlab{F:D:d}%
    (A) The paths realizing the \Frechet distance computed by
    \lemref{FD:d} are stored using an implicit data-structure
    (essentially shortest path trees that are intertwined). This is
    why the space used is linear and why it can be constructed in
    linear time. Of course, an explicit representation of the sequence
    of walks realizing the \Frechet distance might require quadratic
    space in the worst case.

    (B) We emphasize that two consecutive paths of $\PathSet'$, from
    the proof of \lemref{FD:d}, might enclose a region that have
    (potentially) many interior vertices. Thus, the leash might
    ``jump'' over obstacles -- the remainder of this section deals
    with removing this drawback.
\end{remark}

\subsection{Minimum reparameterization width %
   if there are %
   no mountains}

The following lemma implies a $O\pth{\log n}$-approximation algorithm
for the case that all vertices in $\Disk$ are sufficiently close to
both of the two curves.

\begin{lemma}%
    \lemlab{no:tall:h:f:d}%
    Let $\Disk$ be a triangulated topological disk with $n$ faces, and
    $\CTop$ and $\CBot$ be two internally disjoint walks on the
    boundary of $\Disk$. Further, assume for all $\pnt \in \Disk$, the
    distance between $\pnt$ and $\CTop$ is at most $x$, and the
    distance between $\pnt$ and $\CBot$ is at most $x$.  Then one can
    compute reparametrization of $\CBot$ of width $O(x \log n)$.  The
    running time is $O\pth{ n^4 \log n}$ (resp. $O\pth{n^2}$) in the
    continuous (resp. discrete) case.
    
    In particular, if $x =O\pth{ \distFrH{\CTop}{\CBot} }$ then this
    is an $O( \log n)$-approximation to the optimal homotopic \Frechet
    distance.
\end{lemma}

\begin{proof}
    Consider the continuous case.  Using the algorithm of
    \lemref{FD:c} we compute a reparametrization of $\CBot$ of width
    $\delta$, realizing approximately the regular \Frechet distance,
    where $\delta = O(x)$.  Let $\leashX{t}$ denote the leash at time
    $t$ that we obtain from the reparametrization mentioned above.
    Note that the leash $\leashX{\cdot}$ is not required to deform
    continuously in $t$.  In particular, for a given time
    $t \in [0,1]$, let
    $\leashXM{t} = \lim_{t'\rightarrow t^-} \leashX{t'}$ and
    $\leashXP{t} = \lim_{t'\rightarrow t^+} \leashX{t'}$, where
    $\lim_{t'\rightarrow t^-}$ and $\lim_{t'\rightarrow t^+}$ are the
    left-sided and right-sided limits, respectively.  By definition,
    the leash is discontinuous at $t$ if and only if
    $\leashXM{t} \neq \leashXP{t}$.
    
    Naturally, the above reparameterization can be used as long as it
    is continuous. Whenever the leash jumps over a gap (i.e., the
    leash is discontinuous at this point in time), say at time $t$, we
    are going to replace this jump by a
    $\pth{\leashXM{t}, \leashXP{t}}$-homotopy between the two
    leashes. Clearly, this would result in the desired continuous
    homotopy.
    
    To this end, observe that all the vertices inside the disk with
    boundary $\leashXM{t} \concat \leashXP{t}$ have distance $O(x)$ to
    $\CTop$ and $B$, and thus also to $\leashXM{t}$ and
    $\leashXP{t}$. Hence, using the algorithm of \thmref{hh}, compute
    an $\pth{\leashXM{t}, \leashXP{t}}$-homotopy with height
    $O\pth{x \log n}$. Since a gap must contain a vertex there are
    $O(n)$ gaps, so this filling in is done at most $O(n)$ times.
    Computing the initial reparameterization takes
    $O\pth{ n^3 \log n}$ time.  Each gap can be filled in
    $O(n^3 \log n)$ time.
    
    \medskip
    
    The discrete case is similar. The \Frechet distance here can be
    computed in linear time using the algorithm of \lemref{FD:d} (see
    also \remref{F:D:d}).  However, we can only obtain the value of
    the \Frechet distance as well as an implicit representation of the
    actual deformation in linear time. Indeed we can compute an
    explicit listing of the paths in $O(n^2)$ time.  Each path in the
    list can be charged to a single face or edge of $\Disk$.  It
    immediately follows that the number of paths is linear.  For any
    two consecutive paths, $\pi_i$ and $\pi_{i+1}$ in the list, we can
    fill in the possible gap and compute the explicit solution in
    $O(n_i^2)$ time, where $n_i$ is the number of faces between
    $\pi_i$ and $\pi_{i+1}$, see \thmref{discrete} and
    \remref{no:recompute} (D).  Since $\sum{n_i} = O(n)$ the total
    running time of the algorithm is $O(n^2)$.%
    \InDCGVer{\qed} %
\end{proof}

The above lemma demonstrates that if the starting and ending leashes
are known (i.e., the region of the disk $\Disk$ swept over by the
morph) then an approximation algorithm can be obtained. The challenge
is that a priori, we do not know these two leashes, as the input is a
topological disk $\Disk$ with the two curves $\CTop$ and $\CBot$ on
its boundary, and the start/end leashes might be curves that lie
somewhere in the interior of $\Disk$.

\subsection{A Decision Procedure for the Homotopic %
   \Frechet distance in the %
   presence of mountains}

Here, we are handling both the discrete and continuous cases together.

\begin{figure}[t]
    \centerline{%
    \begin{tabular}{ccc}%
      {\IncludeGraphics{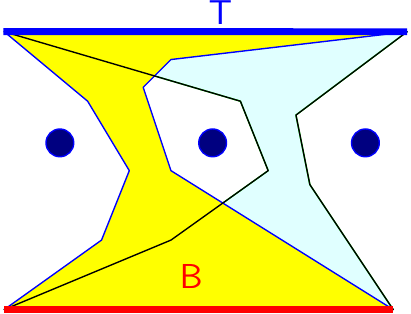}}%
      &%\
        \qquad%
        \qquad%
       \IncludeGraphics{figs/tall}%
        \qquad%
        \qquad%
      &%
        \IncludeGraphics{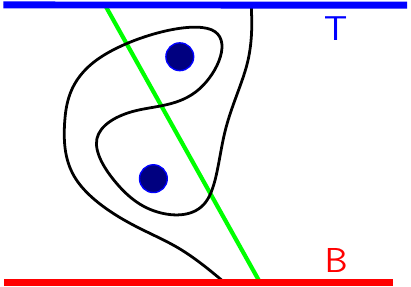}%
      \\%
      (I) & (II) & (III)
    \end{tabular}%
 }
    \caption{}
    \figlab{vallery}%
\end{figure}

For a parameter $\tall \geq 0$, a vertex $v \in \Vertices{\Disk}$ is
\emphi{\tall{}-tall} if its distance to $\CTop$ or $\CBot$ is larger
than $\tall$ (intuitively $\tall$ is a guess for the value of
$\distFrH{\CTop}{\CBot}$).  Here, we consider the case where there are
\tall-tall vertices. Intuitively, one can think about tall vertices as
insurmountable mountains. Thus, to find a good homotopy between
$\CTop$ and $\CBot$, we have to choose which ``valleys'' to use (i.e.,
what homotopy class the solution we compute belongs to if we think
about tall vertices as punctures in the disk). As a concrete example,
consider \figref{vallery}~(I), where there are three tall vertices,
and two possible solutions are being shown.

In the discrete case, we subdivide each edge in the beginning so that
if an edge has length $> 2 \tall$, then the vertex inserted in the
middle of it is \tall-tall.  Observe that, if
$\tall > \distFrH{\CTop}{\CBot}$ then no leash of the optimum
homotopic motion can afford to contain a \tall-tall vertex.  We use
$\TVertices^\tall$ to denote the set of all \tall-tall vertices in
$\Vertices{\Disk}$.

Now, let $\walk$ and $\walk'$ be two walks connecting points on
$\CTop$ and $\CBot$.  The walks $\walk$ and $\walk'$ are
\emphi{homotopic} in $\Disk \setminus \TVertices^\tall$ if and only if
they are homotopic in $\Disk \setminus \TVertices^\tall$ after
contracting $\CTop$ and $\CBot$ (each to a single point).  Two
\emph{non-crossing} walks $\walk$ and $\walk'$ are homotopic if and
only if $\CTop \concat \CBot \concat \walk \concat \walk'$ contains no tall
vertices.  It is straightforward to check that homotopy is an
equivalence relation.  So it partitions $(\CTop, \CBot)$-paths into
\emphi{homotopy classes}; we call each class a $\tau$-homotopy class
or simply a homotopy class (given that $\tau$ is fixed).

For a homotopy class $\iClass$, let $\leftC{\iClass}$
(resp. $\rightC{\iClass}$) be the \emphi{left geodesic}
(resp. \emphi{right geodesic}); that is, $\leftC{\iClass}$ denotes the
shortest path in $\iClass$ from $\vTopL$ to $\vBotL$ (resp. from
$\vTopR$ to $\vBotR$).

Let $\walk$ be any walk in $\iClass$ from $\pBot \in \CBot$ to
$\pTop \in \CTop$.  The \emphi{left tall set} of $\iClass$, denoted by
$\TLeft{\iClass} = \TLeft{\walk}$, is the set of all \tall-tall
vertices to the left of $\walk$. Namely, $\TLeft{\iClass}$ is the set
of tall vertices inside the disk with boundary
\begin{math}
    \CLeft \concat \CTop[\vTopL, \pTop] \concat \walk \concat
    \CBot[\vBotL,\pBot].
\end{math}
where $\CLeft$ is the ``left'' portion of the boundary of $\Disk$,
having endpoints $\vTopL$ and $\vBotL$.  We similarly define the
\emphi{right tall set} of $\iClass$,
$\TRight{\iClass} =\TRight{\walk}$, to be the set of all \tall-tall
vertices to the right of $\walk$.  See \figref{vallery}~(II).

Note that the sets $\TLeft{\iClass}$ and $\TRight{\iClass}$ do not
depend on the particular choice of $\walk$, since all paths in
$\iClass$ are homotopic and so have the same set of $\tall$-tall
vertices to their left and right side.  However, we emphasize that the
left and right tall sets do not identify homotopy classes.
\figref{vallery}~(III) demonstrates two non-homotopic paths with
identical left and right tall sets.

The set $\iClass$ is \emphi{\tall-extendable} from the left if and
only if $\lenX{\leftC{\iClass}} \leq \tall$ and there is a homotopy
class $\iClass'$, such that $\lenX{\leftC{\iClass'}} \leq \tall$ and
$\TLeft{\iClass} \subset \TLeft{\iClass'}$.  In particular, $\iClass$
is \emphi{\tall-saturated} if it is not \tall-extendable and
$\lenX{\leftC{\iClass}} \leq \tau$.

\subsubsection{On the left and right geodesics}

\begin{lemma}%
    \lemlab{saturated:prop}%
    Let $\iClass$ be a $\tall$-saturated homotopy class, where
    $\tall \geq \distFrH{\CTop}{\CBot}$.  Then
    $\lenX{\rightC{\iClass}} \leq 4\tall$.
\end{lemma}

\begin{proof}
    Let $\iClassOpt$ be the homotopy class of the leashes in the
    optimum solution.  Of course, no leash in the optimum solution
    contains a \tall-tall vertex.  Further, all leashes in the optimal
    solution are homotopic because there is a homotopy that contains
    all of them by definition.
    
    Since $\iClass$ is saturated the set $\TLeft{\iClass}$ is not a
    proper subset of $\TLeft{\iClassOpt}$.  It follows that either
    $\TLeft{\iClass} = \TLeft{\iClassOpt}$ or $\TLeft{\iClass}$
    intersects
    $\TVertices^\tall\backslash \TLeft{\iClassOpt} =
    \TRight{\iClassOpt}$
    
    If $\TLeft{\iClass} = \TLeft{\iClassOpt}$ then either
    $\iClass = \iClassOpt$, and in particular
    $\lenX{\rightC{\iClass}} = \lenX{\rightC{\iClassOpt}} \leq \tall$,
    or $\leftC{\iClass}$ crosses $\rightC{\iClassOpt}$.

    \begin{figure}[t]
        \centerline{%
        \begin{tabular}{cc}
          \IncludeGraphics{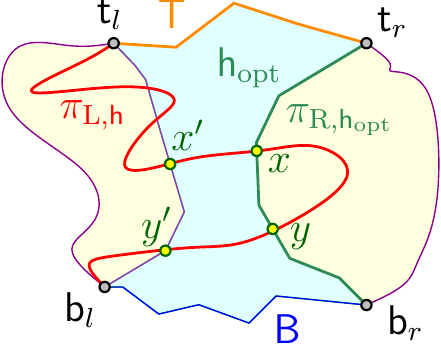}%
          \qquad%
          &
            \qquad\IncludeGraphics[page=2]{figs/frechet_l_r}
          \\
          (I) & (II)
        \end{tabular}%
     }
        \caption{}
        \figlab{l_r}%
    \end{figure}

    Otherwise, the set $\TLeft{\iClass} \cap \TRight{\iClassOpt}$ is
    not empty.  Again, it follows that $\leftC{\iClass}$ crosses
    $\rightC{\iClassOpt}$.
    
    Therefore, we only need to address the case that $\leftC{\iClass}$
    crosses $\rightC{\iClassOpt}$.
    
    Let $x$ be the first intersection point between $\leftC{\iClass}$
    and $\rightC{\iClassOpt}$, as one traverses $\leftC{\iClass}$ from
    $\vTopL$ to $\vBotL$. Let $x'$ be the last intersection point of
    $\leftC{\iClass}[\vTopL, x]$ with $\leftC{\iClassOpt}$.
    Similarly, $y$ is the last intersection point between
    $\leftC{\iClass}$ and $\rightC{\iClassOpt}$, and $y'$ is the first
    intersection of $\leftC{\iClass}[y, \vBotL]$ and
    $\leftC{\iClassOpt}$.  Observe that the interiors of
    $\leftC{\iClass}[x', x]$ and $\leftC{\iClass}[y, y']$ do not
    intersect the curves $\leftC{\iClassOpt}$ and
    $\rightC{\iClassOpt}$. See \figref{l_r} (I).

    As the curves $\leftC{\iClass}$ and $\rightC{\iClass}$ are
    homotopic (by definition), the disk with the boundary
    $\CTop \cdot \leftC{\iClass} \cdot \CBot \cdot \rightC{\iClass}$
    does not contain any tall vertex, and
    $\CTop \cdot \leftC{\iClass} \cdot \CBot$ is homotopic to
    $\rightC{\iClass}$.
    
    Consider the walk
    $\CTop'= \rightC{\iClassOpt}[\vTopR, x] \concat \leftC{\iClass}[x,
    x'] \concat \leftC{\iClassOpt}[x', \vTopL]$, see
    \figref{l_r} (II).
    The walk $\CTop'$ is homotopic to $\CTop$.  Similarly,
    $\CBot' = \leftC{\iClassOpt}[\vBotL, y'] \cdot \leftC{\iClass}[y',
    y] \cdot \rightC{\iClassOpt}[y, \vBotR]$
    is homotopic to $\CBot$.  It follows that $\rightC{\iClass}$ is
    homotopic to $\CTop' \cdot \leftC{\iClass} \cdot \CBot'$. As
    $\rightC{\iClass}$ is the shortest path in its homotopy class with
    these endpoints, it follows that
    \begin{align*}
        \lenX{\rightC{\iClass}} \leq \lenX{\CTop' \cdot
           \leftC{\iClass} \cdot \CBot'}%
        \leq%
        \lenX{\leftC{\iClass} } + \pth{ \lenX{\leftC{\iClassOpt}} +
           \lenX{\leftC{\iClass} } + \lenX{\rightC{\iClassOpt}}}%
        \leq %
        4\tall,
    \end{align*}
    as $\CTop'$ and $\CBot'$ are disjoint, and
    \begin{math}
        \CTop' \cup \CBot'%
        \subseteq%
        \rightC{\iClassOpt} \cup \leftC{\iClassOpt} \cup
        \leftC{\iClass}.
    \end{math}
    \InDCGVer{\qed} %
\end{proof}

A region that contains no $\tall$-tall vertices can still,
potentially, contain $\tall$-tall points (that are not vertices) on
its edges or faces. We next prove that this does not happen in our
setting.

\newcommand{\SQED}{%
   \InDCGVer{\qed}%
   \InNotDCGVer{ %

      \xspace%

   }%
}

\begin{lemma}%
    \lemlab{short:boundaries:c}%
    For any $\tau\geq 0$, let $\iClass$ be a $\tall$-homotopy class,
    such that
    $\max\pth{ \lenX{\leftC{\iClass}}, \lenX{\rightC{\iClass}} } \leq
    x$,
    where $x \geq \tall \geq \distFrH{\CTop}{\CBot}$.  Let $\Disk'$ be
    the disk with boundary
    $\CTop \cdot \rightC{\iClass} \cdot \CBot \cdot \leftC{\iClass}$.
    Then all the points inside $\Disk'$ are within distance $O(x)$ to
    both $\CTop$ and $\CBot$ in $\Disk'$.
\end{lemma}

\begin{proof}
    We first consider the continuous case.  By the definition of
    \tall-homotopy, the disk $\Disk'$ has no \tall-tall
    vertices. Furthermore, by the definition of $x$, we have that the
    distance of any point on $\CTop$ to $\CBot$, restricted to paths
    in $\Disk'$ is at most $\delta_1$, where
    $\delta_1 = x + \distFr{\CTop}{\CBot} \leq 2x$. Indeed, the
    shortest path from any point on $\CTop$ to $\CBot$ in $\Disk$,
    either stays inside $\Disk'$, or alternatively intersects either
    $\leftC{\iClass}$ or $\rightC{\iClass}$.
    
    We can now deploy the decomposition of $\Disk'$ into strips,
    pockets and chunks as done in \secref{decomposition}. Every strip
    (or a chunk) is being swept by a leash of length at most
    $\delta_2 = 2\delta_1 \leq 4x$ (the factor two is because a strip
    might rise out of a delta), and therefore the claim trivially
    holds for points inside such regions.
    
    Every pocket $\Pocket$ has perimeter of length at most
    $\lenX{\bd \Pocket} \leq \delta_3 = 2\delta_2=8x$ (the perimeter
    also contains two points of $\CTop$ and $\CBot$ and they are in
    distance at most $\delta_2$ from each other in either direction
    along the perimeter).  So, consider such a pocket $\Pocket$. Since
    $\Disk'$ contains no \tall-tall vertices, $\Pocket$ does not
    contain any tall vertex. Let $\edge$ be an edge in $\Pocket$ (or a
    subedge if it intersects the boundary of $\Pocket$). The two
    endpoints of $\edge$ are in $\Pocket$, and such an endpoint is
    either a (not tall) vertex or it is contained in $\bd \Pocket$. In
    either case, these endpoints are in distance at most $x$ from
    $\bd \Pocket$, and so they are in distance at most
    $\delta_4 = 2x + \lenX{\bd \Pocket}/2 = 2x + \delta_2 \leq 6x$
    from each other even if the geodesic distance is restricted to
    $P$.  We conclude that $\lenX{\edge} \leq \delta_4$, and
    consequently, any point in $\edge$ is in distance at most
    $\delta_5 = \lenX{\edge}/2 + x + \delta_2 \le 3x + x +8x \leq 12x$
    from $\CTop$ and $\CBot$.
    
    Now, consider any point $\pnt$ in $\Pocket$, and consider the face
    $\face$ that contains it. Since the surface is triangulated,
    $\face$ is a triangle. Clipping $\face$ to $\Pocket$ results in a
    planar region $\face'$ that has perimeter at most
    $\delta_6 = 3\delta_4 + \lenX{\bd \Pocket} \leq 3\cdot 6x +
    \delta_3 \leq (18+8)x \leq 26x$
    (note, that an edge might be fragmented into several subedges, but
    the distance between the furthest two points along a single edge
    is at most $\delta_4$ using the same argument as above).  Thus,
    the furthest a point of $\Pocket$ can be from an edge of $\Pocket$
    is at most $\delta_7 = \delta_6/2\pi \leq 5 x $. Hence, the
    maximum distance of a point of $\Pocket$ from either $\CTop$ or
    $\CBot$ (inside $\Disk'$) is at most
    $\delta_5 + \delta_7 \leq 12x + 5x = 17x$.
    
    \medskip
    
    The discrete case is easy. Any edge of length $\geq 2\tall$ was
    split, by introducing a middle vertex, which must be
    \tall-tall. So the claim immediately holds.%
    \InDCGVer{\qed} %
\end{proof}

\subsubsection{The decision algorithm}

\begin{lemma}%
    \lemlab{try:vertex}%
    Let $\Disk, n, \CTop, \CLeft, \CBot, \CRight, \vTopL, \vBotL$ as
    in the first paragraph of \secref{HFD} and $\tall$ as in the
    previous subsection, and let $X \subseteq \Vertices{\Disk}$ be a
    set of \tall-tall vertices. Consider the shortest path $\PathL$
    (between $\vTopL$ and $\vBotL$) that belongs to any homotopy class
    $\iClass$ such that $X \subseteq \TLeft{\iClass}$. Then the path
    $\PathL$ can be computed in $O\pth{n^4 \log n}$ (resp.
    $O\pth{n \log n}$) time in the continuous (resp. discrete) case.
\end{lemma}

\begin{proof}
    For each vertex of $v \in X$, compute its shortest path $\PathB_v$
    to $\CLeft$ in $\Disk$. Cut the disk $\Disk$ along these
    paths. The result is a topological disk $\Disk'$. Compute the
    shortest path $\PathA$ in $\Disk'$ between $\vTopL$ and $\vBotL$.
    
    We claim that $\PathA=\PathL$. To this end, consider $\PathL$ and
    any path $\PathB_v$ computed by the algorithm. We claim that
    $\PathL$ and $\PathB_v$ do not cross in their interior. Indeed, if
    $\PathL$ cross $\PathB_v$ an odd number of times, then $v$ is
    inside the disk $\PathL \cdot \CTop \cdot \CRight \cdot \CBot$,
    which contradicts the condition that
    $v \in X \subseteq \TLeft{\iClass}$. Clearly, $\PathL$ and
    $\PathB_v$ cannot cross in their interiors more than once, because
    otherwise, one can shorten one of them, which is a contradiction
    as they are both shortest paths. Thus, $\PathL$ is a path in
    $\Disk'$ connecting $\vTopL$ to $\vBotL$, thus implying that
    $\PathA$ is $\PathL$.
    
    As for the running time, each shortest path computation takes time
    $O(n^2 \log n )$, in the continuous (resp. discrete) case. The
    resulting disk has complexity $O\pth{n^2}$, and computing a
    shortest path in it takes $O\pth{n^4 \log n}$ time in the
    continuous case. In the discrete case, computing the paths can be
    done by collapsing $\CLeft$ to a vertex, forbid the shortest path
    tree edges, and run a shortest path algorithm in the remaining
    graph. Clearly, this takes $O(n \log n)$ time.%
    \InDCGVer{\qed} %
\end{proof}

\begin{lemma}%
    \lemlab{saturated:compute}%
    Let $\Disk$ be a triangulated topological disk with $n$ faces, and
    $\CTop$ and $\CBot$ be two internally disjoint walks on $\Disk$'s
    boundary.  Given $\tall > 0$, one can compute a \tall-saturated
    homotopy class, in $O(n^5\log n)$ (resp. $O(n^2\log n)$) time, in
    the continuous (resp.  discrete) case.
\end{lemma}

\begin{proof}
    Start with an empty initial set $X = \emptyset$. At each
    iteration, try adding one of the \tall-tall vertices
    $v \in \TVertices^\tall$ of $\Disk$ to $X$, by using
    \lemref{try:vertex}. The algorithm of \lemref{try:vertex} outputs
    a path $\BPath$ between $\vTopL$ and $\vBotL$ and a set
    $X' \supset X\cup \{v\}$.
    
    If $\BPath$ is of length at most $\tall$ update $X$ to be the new
    set $X'$, otherwise reject $v$.  If $v$ is rejected then the left
    geodesic of any superset of $X\cup \{v\}$ has length larger than
    $\tall$.  It follows that $v$ cannot be accepted in any later
    iteration, so we do not need to reinspect it.  Clearly, after
    trying all the vertices of $\TVertices^\tall$, the set $X$ defines
    the desired saturated class, which can be computed by using the
    algorithm of \lemref{try:vertex}.%
    \InDCGVer{\qed} %
\end{proof}

\begin{lemma}%
    \lemlab{decider}%
    Let $\Disk$ be a triangulated topological disk with $n$ faces, and
    $\CTop$ and $\CBot$ be two internally disjoint walks on the
    boundary of $\Disk$. Given a real number $x > 0$, one can either:
    \begin{compactenum}[\quad(A)]
        \item Compute a homotopy from $\CTop$ to $\CBot$ of width
        $O\pth{ x \log n }$, or
        \item Return that $x < \distFrH{\CTop}{\CBot}$.
    \end{compactenum}
    The running time of this procedure is $O\pth{n^5 \log n }$
    (resp. $O\pth{n^2\log n}$) in the continuous (resp. discrete)
    case.
\end{lemma}

\begin{proof}
    Assume $x \geq \delta_H = \distFrH{\CTop}{\CBot}$, and we use $x$
    as a guess for this value $\delta_H$.  Using
    \lemref{saturated:compute}, one can compute a $x$-saturated
    homotopy class, $\iClass$.  \lemref{saturated:prop} implies that
    both $\leftC{\iClass}$ and $\rightC{\iClass}$ are at most $4x$.
    Let $\Disk' \subseteq \Disk$ be the disk with boundary   
    \begin{math}
        \CTop \concat \leftC{\iClass} \concat \CBot \concat
        \rightC{\iClass}.
    \end{math}
    By \lemref{short:boundaries:c}, all the vertices in $\Disk'$ are
    in distance $O\pth{x}$ from $\CTop$ and $\CBot$ (this holds for
    all points in $\Disk'$ in the continuous case). That is, there are
    no $O(x)$-tall vertices in $\Disk'$.  Finally,
    \lemref{no:tall:h:f:d} implies that a continuous leash sequence of
    height $\leq Z = O\pth{x \log n}$ between $\CTop$ and $B$, inside
    $\Disk'$, can be computed.
    
    Thus, if $x$ is larger than $\distFrH{\CTop}{\CBot}$ then this
    algorithm returns the desired approximation; that is, a homotopy
    of width $\leq Z$. If the width of the generated homotopy is
    however larger than $Z$ (a value that can be computed directly
    from $x$), then the value of $x$ was too small. That is, the
    algorithm fails in this case only if $x <
    \distFrH{\CTop}{\CBot}$.
    In the case of such failure, return that $x$ is too small.%
    \InDCGVer{\qed} %
\end{proof}

\subsection{A strongly polynomial approximation algorithm}%

For a vertex $v \in \Vertices{\Disk}$, define $\CostL{v}$ to be the
length of the shortest path between $\vTopL$ and $\vBotL$ that has $v$
on its left side.  Similarly, for a set of vertices
$X \subseteq \Vertices{\Disk}$, let $\CostSetL{X}$ be the length of
the shortest path between $\vTopL$ and $\vBotL$ that has $X$ on its
left side. For a specific $v$ or $X$, one can compute $\CostL{v}$ and
$\CostSetL{X}$ by invoking the algorithm of \lemref{try:vertex} once.

\subsubsection{The algorithm}

\begin{compactenum}[(I)]
    \item \textbf{Identifying the tall vertices.} %
    Observe that using the algorithm of \lemref{decider}, we can
    decide given a candidate value $\delta_H$ for
    $\distFrH{\CTop}{\CBot}$ if it is too large, too small, or leads
    to the desired approximation.  Indeed, if the algorithm returns an
    approximation of values $O\pth{\delta_H \log n}$ but fails for
    $\delta_H/2$, we know it is the desired approximation.
    
    For each vertex $v \in \Vertices{\Disk}$ let $\alpha_v$ be the
    maximum distance of $v$ to either $\CTop$ or $\CBot$.  Note that
    $v$ cannot be $a$-tall for any $a\geq \alpha_v$.  Sort these
    values, and using binary search, compute the vertex $w$, with the
    minimum value $\alpha_w$, such that \lemref{decider} returns a
    parameterization with homotopic \Frechet distance
    $O( \alpha_w \log n)$. If the algorithm of \lemref{decider}
    returns that $\alpha_w/n$ is too small of a guess, then
    $[\alpha_w/n,\alpha_w \log n]$ contains $\delta_H$.  In this case,
    we can use binary search to find an interval $[\gamma/2, \gamma]$
    that contains $\delta_H$ and use \lemref{decider} to obtain the
    desired approximation.  Similarly, if $v$ is the tallest vertex
    shorter than $w$, then we can assume that $\alpha_v n$ is too
    small of a guess, otherwise we are again done as
    $[\alpha_v,\alpha_v n]$ contains $\delta_H$.
    
    Therefore, in the following, we know that the desired distance
    $\delta_H$ lies in the interval $[x,y]$ where $x=\alpha_v n$ and
    $y=\alpha_w/n$, and for every vertex $u$ of $\Disk$ it holds that
    (i) $\alpha_u \leq x /n$, or (ii) $\alpha_u \geq y n$.  Naturally,
    we consider all the vertices that satisfy (ii) as tall vertices,
    by setting $\tau = 2x/ n$. In the following, let $\TVertices$
    denote the set of these \tall-tall vertices.
    
    \item \textbf{Computing candidate homotopy classes.} %
    Start with $X_0 = \emptyset$. In the $i$\th iteration, the
    algorithm computes the vertex
    $v_i \in\TVertices \setminus X_{i-1}$, such that
    $\CostSetL{X_{i-1} \cup \brc{v_i}}$ is minimized, and set
    $X_{i} = X_{i-1} \cup \brc{v_i}$. Let $\iClass_i$ be the homotopy
    class having $X_i$ on its left side, and
    $\TVertices \setminus X_i$ on its right side.
    
    \item \textbf{Binary search over candidates.}  We approximate the
    homotopic \Frechet width of each one of the classes
    $\iClass_1, \ldots, \iClass_n$. Let $x$ be the minimum homotopic
    \Frechet width computed among these $n$ candidates.
    
    Next, do a binary search in the interval $[x/n^2, x]$ for the
    homotopic \Frechet distance. We return the smallest width
    reparametrization computed as the desired approximation.
\end{compactenum}

\begin{figure}[t]
    \centerline{\IncludeGraphics{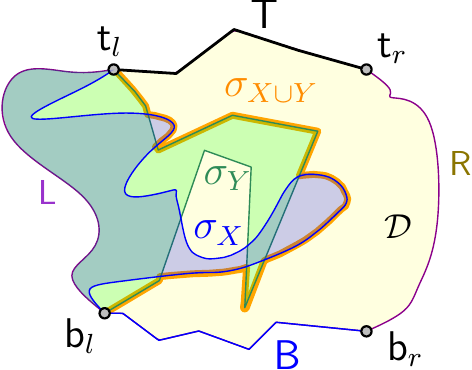}}%
    \caption{}
    \figlab{union}
\end{figure}

\subsubsection{Analysis}

\begin{lemma}%
    \lemlab{union}%
    \begin{inparaenum}[(i)]
        \item For any $X' \subseteq X \subseteq \Vertices{\Disk}$, we
        have $\CostSetL{ X' } \leq \CostSetL{X}$.
        
        \item \itemlab{monotone:single} For any
        $x \in X \subseteq \Vertices{\Disk}$, we have
        $\CostL{ x } \leq \CostSetL{X}$.
        
        \item \itemlab{union} For $X, Y \subseteq \Vertices{\Disk}$,
        we have that
        $\CostSetL{ X \cup Y} \leq \CostSetL{ X } + \CostSetL{Y}$.
    \end{inparaenum}
\end{lemma}

\begin{proof}
    (i) Observe that the path realizing $\CostSetL{X'}$ is less
    constrained than the path realizing $\CostSetL{X}$, therefore it
    might only be shorter.

    (ii) Follows immediately from (i).
    
    \smallskip
    
    (iii) Consider the disk $\Disk$ and the two paths $\BPath_X$ and
    $\BPath_Y$ realizing $\CostSetL{ X }$ and $\CostSetL{Y}$,
    respectively.  The close curves $\BPath_x \concat \CLeft$ and
    $\BPath_Y \concat \CLeft$ encloses two topological disks. Consider
    the union of these two disks, and its connected outer boundary
    $\BPath_{X \cup Y} \cup \CLeft$. Clearly, $\BPath_{X \cup Y}$
    connects $\vTopL$ and $\vBotL$, and it has all the points of $X$
    and $Y$ on one side of it, and finally
    $\lenX{\BPath_{X \cup Y}} \leq \lenX{\BPath_X} + \lenX{\BPath_Y}$
    as $\BPath_{X \cup Y} \subseteq \BPath_X \cup \BPath_Y$. See 
    \figref{union}.
\end{proof}

\begin{lemma}%
    \lemlab{range}%
    The cheapest homotopic \Frechet parameterization computed among
    $\iClass_1, \ldots, \iClass_n$ has width
    $O \pth{ \distFrH{\CTop}{\CBot} n \log n }$.
\end{lemma}
\begin{proof}
    Consider the set $Y$ that is the subset of tall vertices on the
    left side of the optimal solution. Let $i$ be the first index such
    that $Y \subseteq X_i$ and $Y \not \subseteq X_{i-1}$. Let $v$ be
    any vertex in $Y \setminus X_{i-1}$.  By construction, we have
    that $\CostSetL{ X_{i} } \leq \CostSetL{ X_{i-1} \cup \brc{v}}$,
    and furthermore, for all $j \leq i$, we have that
    $\CostSetL{ X_{j} } \leq \CostSetL{ X_{j-1} \cup \brc{v}}$, by the
    greediness in the construction of $X_1, \ldots, X_i$.  Now, we
    have
    \begin{align*}
        \CostSetL{ X_{i} }%
        &\leq%
        \CostSetL{ X_{i-1} \cup \brc{v}}%
        & \text{(by construction of $X_i$)}%
        \\&\leq%
        \CostSetL{ X_{i-1} } + \CostL{v} %
        & \text{(by \lemref{union} \itemref{union})}%
        \\%
        &\leq %
        \CostSetL{ X_{i-1} } + \CostSetL{Y}%
        &%
        \text{(by \lemref{union} \itemref{monotone:single})}%
        \\%
        &%
        \leq%
        \pth{\CostSetL{ X_{i-2} } + \CostSetL{Y}} + \CostSetL{Y}%
        &%
        \text{(applying same argument to $X_{i-1}$)}
        \\
        &%
        =%
        \CostSetL{ X_{i-2} } + 2\CostSetL{Y}%
        \\&%
        \leq \cdots \leq%
        i\CostSetL{Y}%
        \leq%
        n \CostSetL{Y}.%
    \end{align*}
    Now, setting $\tall = \CostSetL{ X_{i} }$, it follows that $X_i$
    is $\tall$-saturated. Applying \lemref{saturated:prop}, implies
    that $\lenX{\rightC{\iClass_i}} \leq 4\tall$.  Observe, that the
    disk defined by $\CTop$, $\leftC{\iClass_i}$, $\CBot$,
    $\rightC{\iClass_i}$ cannot contain any tall vertex (by
    construction).
    
    Now, plugging this into \lemref{no:tall:h:f:d} implies the
    homotopic \Frechet width of $\iClass_i$ (starting with
    $\leftC{\iClass_i}$ and ending up with $\rightC{\iClass_i}$, so
    $\Disk$ in \lemref{no:tall:h:f:d} is bounded by
    $\CTop,\CBot,\leftC{\iClass_i},\rightC{\iClass_i}$) is
    $O\pth{ \tall \log n}$, which implies the claim since
    $\CostSetL{X_i}\leq n\CostSetL{Y}\leq n\distFrH{\CTop}{\CBot}$.
    \InDCGVer{\qed} %
\end{proof}

\subsubsection{The algorithm}

\begin{theorem}%
    \thmlab{h:f:d:main}%
    Let $\Disk$ be a triangulated topological disk with $n$ faces, and
    $\CTop$ and $\CBot$ be two internally disjoint walks on the
    boundary of $\Disk$. One can compute a homotopic \Frechet
    parameterization of $\CTop$ and $\CBot$ of width
    $O\pth{ \distFrH{\CTop}{\CBot} \log n}$, where
    $\distFrH{\CTop}{\CBot}$ is the homotopic \Frechet distance
    between $\CTop$ and $\CBot$ in $\Disk$.
    
    The running time of this procedure is $O\pth{n^6 \log n }$
    (resp. $O\pth{n^3\log n}$) in the continuous (resp. discrete)
    case.
\end{theorem}

\begin{proof}
    Consider the algorithm described in the previous subsection.  For
    correctness, observe that the algorithm either found the desired
    value, or identified correctly the tall vertices. Next, by
    \lemref{range}, the range the algorithm searches over contains the
    desired value.
    
    The algorithm requires $O(n^2)$ calls to \lemref{try:vertex},
    which takes $O\pth{n^6 \log n}$ (resp. $O\pth{n^3 \log n}$) time
    in the continuous (resp. discrete) case.  Then the algorithm
    requires the method of \lemref{no:tall:h:f:d} to compute the
    homotopic \Frechet distance of the classes
    $\iClass_1, \ldots, \iClass_n$. The algorithm also performs
    $O( \log n)$ calls to the algorithm of \lemref{decider}.
    \InDCGVer{\qed} %
\end{proof}

\section{Conclusions}
\seclab{conclusions}

We presented a $O( \log n)$ approximation algorithm for approximating
the homotopy height and the homotopic \Frechet distance between curves
on piecewise linear surfaces. It seems quite believable that the
approximation quality can be further improved, and we leave this as
the main open problem. Since our algorithm works both for the
continuous and discrete cases, it seems natural to conjecture that
this algorithm should also work for more general surfaces and metrics.

Another problem for further research is to solve our main problem
without the restriction that the two curves lie on the boundary of the
disk.

% \begin{remark}
\paragraph{Connection to planar separator.}
Our basic algorithm (\thmref{discrete}) is inspired to some extent by
the proof of the planar separator theorem \cite{lt-stpg-79}. In
particular, our result implies sufficient conditions to having a
separator that can continuously deform from enclosing nothing in a
planar graph, till it encloses the whole graph, without being too long
at any point in time. As a result, our work can be viewed as extending
the planar separator theorem.  A natural open problem is to extend our
work to graphs with higher genus.
% \end{remark}

\paragraph{Acknowledgments}
The authors thank Jeff Erickson and Gary Miller for their comments and
suggestions. The authors also thank the anonymous referees for their
detailed and insightful reviews.

% -------------------------------------------------------------------------

\InDCGVer{%
%*flatex input: [./homotopy.bbl]
\newcommand{\etalchar}[1]{$^{#1}$}
 \providecommand{\CNFX}[1]{ {\em{\textrm{(#1)}}}}
  \providecommand{\tildegen}{{\protect\raisebox{-0.1cm}{\symbol{'176}\hspace{-0.03cm}}}}
  \providecommand{\SarielWWWPapersAddr}{http://sarielhp.org/p/}
  \providecommand{\SarielWWWPapers}{http://sarielhp.org/p/}
  \providecommand{\urlSarielPaper}[1]{\href{\SarielWWWPapersAddr/#1}{\SarielWWWPapers{}/#1}}
  \providecommand{\Badoiu}{B\u{a}doiu}
  \providecommand{\Barany}{B{\'a}r{\'a}ny}
  \providecommand{\Bronimman}{Br{\"o}nnimann}  \providecommand{\Erdos}{Erd{\H
  o}s}  \providecommand{\Gartner}{G{\"a}rtner}
  \providecommand{\Matousek}{Matou{\v s}ek}
  \providecommand{\Merigot}{M{\'{}e}rigot}
  \providecommand{\Hastad}{H\r{a}stad\xspace}  \providecommand{\CNFSoCG}{
  {\em{\textrm{(SoCG)}}}} \providecommand{\CNFCCCG}{\CNFX{CCCG}}
  \providecommand{\CNFFOCS}{\CNFX{FOCS}}
  \providecommand{\CNFSODA}{\CNFX{SODA}}  \providecommand{\CNFSTOC}{
  {\em{\textrm{(STOC)}}}}  \providecommand{\CNFBROADNETS}{\CNFX{BROADNETS}}
  \providecommand{\CNFESA}{\CNFX{ESA}}
  \providecommand{\CNFFSTTCS}{\CNFX{FSTTCS}}
  \providecommand{\CNFIJCAI}{\CNFX{IJCAI}}
  \providecommand{\CNFINFOCOM}{\CNFX{INFOCOM}}
  \providecommand{\CNFIPCO}{\CNFX{IPCO}}
  \providecommand{\CNFISAAC}{\CNFX{ISAAC}}
  \providecommand{\CNFLICS}{\CNFX{LICS}}
  \providecommand{\CNFPODS}{\CNFX{PODS}}
  \providecommand{\CNFSWAT}{\CNFX{SWAT}}
  \providecommand{\CNFWADS}{\CNFX{WADS}}

% flatex input end: [./homotopy.bbl]
}%
\InNotDCGVer{%
%*flatex input: [./homotopy.bbl]
\newcommand{\etalchar}[1]{$^{#1}$}
 \providecommand{\CNFX}[1]{ {\em{\textrm{(#1)}}}}
  \providecommand{\tildegen}{{\protect\raisebox{-0.1cm}{\symbol{'176}\hspace{-0.03cm}}}}
  \providecommand{\SarielWWWPapersAddr}{http://sarielhp.org/p/}
  \providecommand{\SarielWWWPapers}{http://sarielhp.org/p/}
  \providecommand{\urlSarielPaper}[1]{\href{\SarielWWWPapersAddr/#1}{\SarielWWWPapers{}/#1}}
  \providecommand{\Badoiu}{B\u{a}doiu}
  \providecommand{\Barany}{B{\'a}r{\'a}ny}
  \providecommand{\Bronimman}{Br{\"o}nnimann}  \providecommand{\Erdos}{Erd{\H
  o}s}  \providecommand{\Gartner}{G{\"a}rtner}
  \providecommand{\Matousek}{Matou{\v s}ek}
  \providecommand{\Merigot}{M{\'{}e}rigot}
  \providecommand{\Hastad}{H\r{a}stad\xspace}  \providecommand{\CNFSoCG}{
  {\em{\textrm{(SoCG)}}}} \providecommand{\CNFCCCG}{\CNFX{CCCG}}
  \providecommand{\CNFFOCS}{\CNFX{FOCS}}
  \providecommand{\CNFSODA}{\CNFX{SODA}}  \providecommand{\CNFSTOC}{
  {\em{\textrm{(STOC)}}}}  \providecommand{\CNFBROADNETS}{\CNFX{BROADNETS}}
  \providecommand{\CNFESA}{\CNFX{ESA}}
  \providecommand{\CNFFSTTCS}{\CNFX{FSTTCS}}
  \providecommand{\CNFIJCAI}{\CNFX{IJCAI}}
  \providecommand{\CNFINFOCOM}{\CNFX{INFOCOM}}
  \providecommand{\CNFIPCO}{\CNFX{IPCO}}
  \providecommand{\CNFISAAC}{\CNFX{ISAAC}}
  \providecommand{\CNFLICS}{\CNFX{LICS}}
  \providecommand{\CNFPODS}{\CNFX{PODS}}
  \providecommand{\CNFSWAT}{\CNFX{SWAT}}
  \providecommand{\CNFWADS}{\CNFX{WADS}}

% flatex input end: [./homotopy.bbl]
}%
%FLATEX-REM:\bibliography{homotopy}%
%\bibliography{shortcuts,geometry}%

\begin{thebibliography}{BBG{\etalchar{+}}08b}

\bibitem[AAOS97]{aaos-supa-97}
\href{http://www.cs.duke.edu/~pankaj}{P.~K.~{Agarwal}}, \href{http://cis.poly.edu/~aronov/}{B.~{Aronov}}, \href{http://cs.smith.edu/~orourke/}{J.~{O'Rourke}}, and C.~A. Schevon.
\newblock  Star unfolding of a polytope with applications.
\newblock {\em SIAM J. Comput.}, 26:1679--1713, 1997.

\bibitem[AB05]{ab-scfds-05}
{\href{http://www.inf.fu-berlin.de/inst/ag-ti/members/alt.en.html}{H.~{Alt}}} and M.~Buchin.
\newblock  Semi-computability of the {\Frechet}~distance between surfaces.
\newblock In {\em Proc. 21st Euro. Workshop on Comput. Geom.}, pages 45--48,
  2005.

\bibitem[AG95]{ag-cfdbt-95}
{\href{http://www.inf.fu-berlin.de/inst/ag-ti/members/alt.en.html}{H.~{Alt}}} and M.~Godau.
\newblock  Computing the {Fr\'echet} distance between two polygonal curves.
\newblock {\em Internat. J. Comput. Geom. Appl.}, 5:75--91, 1995.

\bibitem[BBG08a]{bbg-dsfm-08}
K.~Buchin, M.~Buchin, and J.~Gudmundsson.
\newblock  Detecting single file movement.
\newblock In {\em Proc. 16th ACM SIGSPATIAL Int. Conf. Adv. GIS}, pages
  288--297, 2008.

\bibitem[BBG{\etalchar{+}}08b]{bbgll-dcpcs-08}
K.~Buchin, M.~Buchin, J.~Gudmundsson, Maarten L., and J.~Luo.
\newblock  Detecting commuting patterns by clustering subtrajectories.
\newblock In {\em Proc. 19th Annu. Internat. Sympos. Algorithms
  Comput.\CNFISAAC}, pages 644--655, 2008.

\bibitem[BDS13]{bds-cfdsn-13}
M.~Buchin, A.~Driemel, and B.~Speckmann.
\newblock  Computing the {Fr{\'e}chet} distance with shortcuts is np-hard.
\newblock {\em CoRR}, abs/1307.2097, 2013.

\bibitem[BPSW05]{bpsw-mmvtd-05}
S.~Brakatsoulas, D.~Pfoser, R.~Salas, and \href{http://www.cs.utsa.edu/~carola/}{C.~{Wenk}}.
\newblock  On map-matching vehicle tracking data.
\newblock In {\em Proc. 31st VLDB Conference}, pages 853--864. VLDB Endowment,
  2005.

\bibitem[BVIG91]{bvig-psfnd-91}
C.~Bennis, J.-M. V{\'e}zien, G.~Igl{\'e}sias, and A.~Gagalowicz.
\newblock  Piecewise surface flattening for non-distorted texture mapping.
\newblock In Thomas~W. Sederberg, editor, {\em Proc. SIGGRAPH '91}, volume~25,
  pages 237--246, 1991.

\bibitem[BW09]{bw-sp-09}
G.~A. Brightwell and P.~Winkler.
\newblock \href{http://dx.doi.org/10.1137/07069078X}{Submodular percolation}.
\newblock {\em SIAM J. Discret. Math.}, 23(3):1149--1178, 2009.

\bibitem[CCE{\etalchar{+}}10]{cdellt-wydwpt-10}
E.~W. Chambers, E.~{Colin de Verdi{\`e}re}, \href{http://compgeom.cs.uiuc.edu/~jeffe/}{J.~{Erickson}}, S.~Lazard, F.~Lazarus,
  and S.~Thite.
\newblock  Homotopic fr{\'e}chet distance between curves or, walking your dog
  in the woods in polynomial time.
\newblock {\em Comput. Geom. Theory Appl.}, 43(3):295--311, 2010.

\bibitem[CDH{\etalchar{+}}11]{cdhsw-cfdfp-11}
A.~F. Cook, A.~Driemel, \href{http://sarielhp.org}{S.~{{Har-Peled}}}, J.~Sherette, and \href{http://www.cs.utsa.edu/~carola/}{C.~{Wenk}}.
\newblock  Computing the {Fr{\'e}chet} distance between folded polygons.
\newblock In {\em Proc. 12th Workshop Algorithms Data Struct.\CNFWADS}, pages
  267--278, 2011.

\bibitem[CL09]{cl-hh-09}
E.~W. Chambers and D.~Letscher.
\newblock  On the height of a homotopy.
\newblock In {\em Proc. 21st Canad. Conf. Comput. Geom.\CNFCCCG}, 2009.

\bibitem[CL10]{cl-ehh-10}
E.~W. Chambers and D.~Letscher.
\newblock \href{http://mathcs.slu.edu/~chambers/papers/hherratum.pdf}{Erratum
  for on the height of a homotopy}.
\newblock \url{http://mathcs.slu.edu/~chambers/papers/hherratum.pdf}, 2010.

\bibitem[CLJL11]{cljl-ifd-11}
E.~W. Chambers, D.~Letscher, T.~Ju, and L.~Liu.
\newblock  Isotopic {Fr{\'e}chet} distance.
\newblock In {\em Proc. 23rd Canad. Conf. Comput. Geom.\CNFCCCG}, 2011.

\bibitem[CR13]{cr-clr2s-13}
G.~R. {Chambers} and R.~{Rotman}.
\newblock  {Contracting loops on a {Riemannian} $2$-surface}.
\newblock {\em ArXiv e-prints}, November 2013.

\bibitem[CW10]{cw-gfdis-10}
A.~F. Cook and \href{http://www.cs.utsa.edu/~carola/}{C.~{Wenk}}.
\newblock  Geodesic {Fr{\'e}chet} distance inside a simple polygon.
\newblock {\em ACM Trans. Algo.}, 7:9:1--9:19, 2010.

\bibitem[CW12]{cw-sppps-12}
A.~F. Cook and \href{http://www.cs.utsa.edu/~carola/}{C.~{Wenk}}.
\newblock  Shortest path problems on a polyhedral surface.
\newblock {\em Algorithmica}, 2012.
\newblock to appear.

\bibitem[CW13]{cw-msbc2-13}
E.~W. Chambers and \href{http://www.cs.duke.edu/~wys/}{Y. {Wang}}.
\newblock \href{http://doi.acm.org/10.1145/2462356.2462375}{Measuring
  similarity between curves on $2$-manifolds via homotopy area}.
\newblock In {\em Proc. 29th Annu. Sympos. Comput. Geom.\CNFSoCG}, pages
  425--434, 2013.

\bibitem[DHW12]{dhw-afdrc-12}
A.~Driemel, \href{http://sarielhp.org}{S.~{{Har-Peled}}}, and \href{http://www.cs.utsa.edu/~carola/}{C.~{Wenk}}.
\newblock  Approximating the {Fr\'{e}chet} distance for realistic curves in
  near linear time.
\newblock {\em \href{http://link.springer.com/journal/454}{Discrete Comput. {}Geom.}}, 48:94--127, 2012.

\bibitem[EGH{\etalchar{+}}02]{eghmm-nsmpa-02}
\href{http://www.cs.arizona.edu/~alon/}{A.~{Efrat}}, \href{http://geometry.stanford.edu/member/guibas/}{L.~J.~Guibas}, \href{http://sarielhp.org}{S.~{{Har-Peled}}}, J.~S.B. Mitchell, and T.M. Murali.
\newblock  New similarity measures between polylines with applications to
  morphing and polygon sweeping.
\newblock {\em \href{http://link.springer.com/journal/454}{Discrete Comput. {}Geom.}}, 28:535--569, 2002.

\bibitem[EM94]{em-cdfd-94}
T.~Eiter and H.~Mannila.
\newblock  Computing discrete {\Frechet{}} distance.
\newblock Tech. Report CD-TR 94/64, Christian Doppler Lab. Expert Sys., TU
  Vienna, Austria, 1994.

\bibitem[Flo97]{f-psast-97}
M.~S. Floater.
\newblock  Parameterization and smooth approximation of surface triangulations.
\newblock {\em Comput. Aided Geom. Design}, 14(4):231--250, 1997.

\bibitem[Fre24]{f-sdds-24}
M.~Frech\'{e}t.
\newblock  Sur la distance de deux surfaces.
\newblock {\em Ann.~Soc.~Polonaise Math.}, 3:4--19, 1924.

\bibitem[God99]{g-cmsbg-99}
M.~Godau.
\newblock  {\em On the complexity of measuring the similarity between geometric
  objects in higher dimensions}.
\newblock PhD thesis, Free University of Berlin, 1999.

\bibitem[HKRS97]{hkrs-fspap-97}
M.~R. Henzinger, P.~Klein, S.~Rao, and S.~Subramanian.
\newblock \href{http://portal.acm.org/citation.cfm?id=261540.261541}{Faster
  shortest-path algorithms for planar graphs}.
\newblock {\em J. Comput. Sys. Sci.}, 55:3--23, August 1997.

\bibitem[HNSS12]{hnss-hwydm-12}
\href{http://sarielhp.org}{S.~{{Har-Peled}}}, A.~Nayyeri, M.~Salavatipour, and A.~Sidiropoulos.
\newblock  How to walk your dog in the mountains with no magic leash.
\newblock In {\em Proc. 28th Annu. Sympos. Comput. Geom.\CNFSoCG}, pages
  121--130, 2012.

\bibitem[HR11]{hr-fdre-11}
\href{http://sarielhp.org}{S.~{{Har-Peled}}} and B.~Raichel.
\newblock  The {Fr\'{e}chet} distance revisited and extended.
\newblock In {\em Proc. 27th Annu. Sympos. Comput. Geom.\CNFSoCG}, pages
  448--457, New York, NY, USA, 2011. ACM.
\newblock \url{http://sarielhp.org/papers/10/frechet3d/}.

\bibitem[KKS05]{kks-osmut-05}
M.S. Kim, S.W. Kim, and M.~Shin.
\newblock  Optimization of subsequence matching under time warping in
  time-series databases.
\newblock In {\em Proc. ACM symp. Applied comput.}, pages 581--586, 2005.

\bibitem[KP99]{kp-sudtw-99}
E.~J. Keogh and M.~J. Pazzani.
\newblock  Scaling up dynamic time warping to massive dataset.
\newblock In {\em Proc. of the Third Euro. Conf. Princip. Data Mining and Know.
  Disc.}, pages 1--11, 1999.

\bibitem[LSH65]{lsh-mbrcf-65}
P.~M. {Lewis II}, R.~E. Stearns, and J.~Hartmanis.
\newblock \href{http://dx.doi.org/10.1109/FOCS.1965.14}{Memory bounds for
  recognition of context-free and context-sensitive languages}.
\newblock In {\em Proc. 6th Annu. IEEE Sympos. Found. Comput. Sci.\CNFFOCS},
  pages 191--202, 1965.

\bibitem[LT79]{lt-stpg-79}
R.~J. Lipton and R.~E. Tarjan.
\newblock  A separator theorem for planar graphs.
\newblock {\em SIAM J. Appl. Math.}, 36:177--189, 1979.

\bibitem[MDBH06]{mdbh-cmpdfd-06}
A.~Mascret, T.~Devogele, I.~Le Berre, and A.~H\'{e}naff.
\newblock  Coastline matching process based on the discrete {Fr\'{e}chet}
  distance.
\newblock In {\em Proc. 12th Int. Sym. Spatial Data Handling}, pages 383--400,
  2006.

\bibitem[MMP87]{mmp-dgp-87}
J.~S.B. Mitchell, \href{http://www.cs.umd.edu/~mount/}{D.~M. {Mount}}, and C.~H. Papadimitriou.
\newblock  {The discrete geodesic problem}.
\newblock {\em SIAM J. Comput.}, 16:647--668, 1987.

\bibitem[Pap13]{p-ctd-13}
P.~Papasoglu.
\newblock  Contracting thin disks.
\newblock {\em ArXiv e-prints}, September 2013.

\bibitem[PB00]{pb-stmss-00}
D.~Piponi and G.~Borshukov.
\newblock  Seamless texture mapping of subdivision surfaces by model pelting
  and texture blending.
\newblock In {\em Proc. SIGGRAPH 2000}, pages 471--478, August 2000.

\bibitem[SdS00]{ss-spmtf-00}
A.~Sheffer and E.~de~Sturler.
\newblock  Surface parameterization for meshing by triangulation flattening.
\newblock In {\em Proc. 9th International Meshing Roundtable}, pages 161--172,
  2000.

\bibitem[SGHS08]{sgh-csi-08}
J.~Serr{\`a}, E.~G{\'o}mez, P.~Herrera, and X.~Serra.
\newblock  Chroma binary similarity and local alignment applied to cover song
  identification.
\newblock {\em IEEE Transactions on Audio, Speech {\&} Language Processing},
  16(6):1138--1151, 2008.

\bibitem[SW13]{sw-sce-13}
J.~Sherette and \href{http://www.cs.utsa.edu/~carola/}{C.~{Wenk}}.
\newblock  Simple curve embedding.
\newblock {\em CoRR}, abs/1303.0821, 2013.

\bibitem[WSP06]{wsp-anmms-06}
\href{http://www.cs.utsa.edu/~carola/}{C.~{Wenk}}, R.~Salas, and D.~Pfoser.
\newblock  Addressing the need for map-matching speed: Localizing global
  curve-matching algorithms.
\newblock In {\em Proc. 18th Int. Conf. Sci. Statis. Database Manag.}, pages
  879--888, 2006.

\end{thebibliography}

\appendix

\section{Sweeping a convex polytope, star %
   unfolding, and banana peels}
\apndlab{sweep:p}

%\medskip%

Consider a convex polytope $\Polytope$ in three dimensions, a base
point $\bpnt$ on its boundary, and the problem of finding the minimum
length leash needed for a guard that walks on the polytope such that
the leash sweeps over all the points on the surface of the
polytope. Specifically, at any point in time, the guard maintains a
connection to the base point $\bpnt$ via a path (i.e., the leash)
connecting it to the base point, and the leash has to move
continuously as the guard moves around.

For a point $\pnt$ on the boundary of $\Polytope$, let $\distP{\pnt}$
be the geodesic distance from $\bpnt$ to $\pnt$ (i.e., the length
shortest path $\PathA$ that lies on the boundary of $\Polytope$
connecting $\bpnt$ to $\pnt$). Let $\MAxis$ be the \emphi{medial axis}
of this distance -- formally, a point $\pnt$ is on the medial axis if
there are two distinct shortest paths $\PathA$ and $\PathB$ from
$\bpnt$ to $\pnt$, such that
\begin{math}
    \lenX{\PathA} = \lenX{\PathB} = \distP{\pnt}.
\end{math}
It is known that $\MAxis$ is a tree in this case \cite{aaos-supa-97}.

Now, let $\PathSet$ be the union of all the shortest paths from
$\bpnt$ to the vertices of $\Polytope$ (we assume that no vertex is on
the medial axis, which holds under general position assumption). The
set $\PathSet$ is also a tree. Surprisingly, if you cut
$\bd{\Polytope}$ along $\PathSet$, then the resulting polygon can be
flattened on the plane. Maybe even more surprisingly, this even holds
if one cuts $\bd{\Polytope}$ along $\MAxis$. This is known as
\emph{star unfolding} of a polytope, see Agarwal
\etal~\cite{aaos-supa-97} for details.

Consider cutting $\bd{\Polytope}$ along both $\MAxis$ and
$\PathSet$. This breaks $\Polytope$ into a collection of polygons
$\PolySet$, where each polygon $\Polygon \in \PolySet$, has no
vertices of $\Polytope$ in its interior, and has $\bpnt$ as a
vertex. As such, one can unfold this $\Polygon$ into the plane.  Here,
the two paths of $\PathSet$ adjacent to $\bpnt$ that belong to the
boundary of $\Polygon$ maps in this unfolding to two straight
edges. The rest of the boundary $\Polygon$ is a closed connected
portion of $\MAxis$.  One can think about $\Polygon$ as being a
``leaf'' in a decomposition of $\bd{\Polytope}$ (i.e., think about the
sides of a banana peel). Here, shortest paths from $\bpnt$ to any
point on $\pnt \in \bd{\Polytope}$ that belongs to $\Polygon$ results
in a straight segment in (the planar embedded version of) the polygon
$\Polygon$. As such, the polygons of $\PolySet$ completely capture the
structure of all the shortest paths on $\bd{\Polytope}$ to $\bpnt$.

Back to the problem of sweeping $\bd \Polytope$. For the points of
$\MAxis_\Polygon = \MAxis \cap \bd \Polygon$, we can sweep the region
of $\bd \Polygon$ that corresponds to $\Polygon$, by walking along the
curve $\MAxis_\Polygon$ (say counterclockwise), and the leash being
the shortest path in $\bd \Polytope$ (that lies inside
$\Polygon$). This completely sweeps over the region of $\Polygon$. We
then continue this sweeping in the next polygon of $\PolySet$ adjacent
to $\Polygon$ around $\bpnt$. We continue in this fashion till all the
boundary of the polytope is swept over. Note, that the leash is moving
continuously, and during this motion, the maximum length of the leash
is the distance to the furthest point on $\bd \Polytope$ from
$\bpnt$. We conclude that this is an optimal solution and using the
known algorithms for computing shortest paths \cite{aaos-supa-97}.

Let us recap the algorithm: We compute the medial axis $\MAxis$ of
$\bpnt$ on $\bd \Polytope$, under the shortest path distance on the
boundary of the polytope. Next, the parameterize a point $\pnt(t)$ to
move continuously around the tree $\MAxis$ (i.e., traversing along
each edge twice, in both direction).  At each point in time, the leash
is connected via the shortest path to the base point $\bpnt$.

\begin{lemma}
    Given a convex polytope $\Polytope$ in three dimensions, and a
    base point $\pnt \in \bd \Polytope$, one can compute in polynomial
    time, a continuous motion of a point $\pnt(t)$, $t \in [0,1]$, and
    an associated leash $\leashX{t}$ connecting $\pnt(t)$ with
    $\bpnt$, such that
    \begin{inparaenum}[(i)]
        \item the leash sweeps over all the points of $\bd \Polytope$,
        \item the leash moves continuously,
        \item a point of $\bd \Polytope$ get swept over only once
        during this motion, and
        \item the maximum length of the leash is
        $\max_{ \pnt \in \bd \Polytope} \distP{\pnt}$, which is
        optimal.
    \end{inparaenum}
\end{lemma}

The above is in sharp contrast to our original problem of computing
the homotopy height, as the two ends of the leash must move along two
prespecified curves $\CLeft$ and $\CRight$. Furthermore, because of
that we no longer have the property that the leashes do not jump, as
the leash head no longer moves along the medial axis, as the resulting
paths might be too long (compared to the optimal morph). Nevertheless,
the above captures our basic strategy of breaking the input disk into
smaller slivers, induced by shortest paths, and solving the problem in
each sliver separately, and gluing the solutions together.

\end{document}